\batchmode
\makeatletter
\def\input@path{{/Users/bgraham/Dropbox/Research/Causal_Inference/ProgramEvaluation/Average_Regression/Writing/}}
\makeatother
\documentclass[12pt,english]{article}
\usepackage[T1]{fontenc}
\usepackage[latin9]{inputenc}
\usepackage{geometry}
\usepackage{babel}
\usepackage{float}
\usepackage{url}
\usepackage{amsmath}
\usepackage{amsthm}
\usepackage{amssymb}
\usepackage{stmaryrd}
\usepackage{setspace}
\usepackage[authoryear]{natbib}
\PassOptionsToPackage{normalem}{ulem}
\usepackage{ulem}
\usepackage[unicode=true]
 {hyperref}
\usepackage{breakurl}

\makeatletter

\providecommand{\tabularnewline}{\\}
\floatstyle{ruled}
\newfloat{algorithm}{tbp}{loa}
\providecommand{\algorithmname}{Algorithm}
\floatname{algorithm}{\protect\algorithmname}

  \theoremstyle{plain}
  \newtheorem{assumption}{\protect\assumptionname}
  \theoremstyle{definition}
  \newtheorem{defn}{\protect\definitionname}
  \theoremstyle{plain}
  \newtheorem{lem}{\protect\lemmaname}
  \theoremstyle{plain}
  \newtheorem{prop}{\protect\propositionname}
\theoremstyle{plain}
\newtheorem{thm}{\protect\theoremname}
  \theoremstyle{plain}
  \newtheorem{cor}{\protect\corollaryname}

\usepackage{pdflscape}

\makeatother

  \providecommand{\assumptionname}{Assumption}
  \providecommand{\definitionname}{Definition}
  \providecommand{\lemmaname}{Lemma}
  \providecommand{\propositionname}{Proposition}
\providecommand{\corollaryname}{Corollary}
\providecommand{\theoremname}{Theorem}

\begin{document}
\begin{singlespacing}

\title{Semiparametrically efficient estimation of the average linear regression
function}
\maketitle
\begin{center}
Bryan S. Graham$^{+}$ and Cristine Campos de Xavier Pinto$^{\#}$\footnote{{\footnotesize{}$^{+}$Department of Economics, University of California
- Berkeley, 530 Evans Hall \#3380, Berkeley, CA 94720-3888 and National
Bureau of Economic Research, }{\footnotesize{}\uline{e-mail:}}{\footnotesize{}
\href{http://bgraham@econ.berkeley.edu}{bgraham@econ.berkeley.edu},
}{\footnotesize{}\uline{web:}}{\footnotesize{} \url{http://bryangraham.github.io/econometrics/}.
}\\
{\footnotesize{}$^{\#}$Escola de Economia de Sao Paulo, FGV, Rua
Itapeva 474, sala 1010, CEP: 01332-000. }{\footnotesize{}\uline{e-mail:}}{\footnotesize{}
$\mathtt{cristinepinto@gmail.com}$. }{\footnotesize{}\uline{web:}}{\footnotesize{}
$\mathtt{http://sites.google.com/site/cristinepinto/}$.}\\
{\footnotesize{}We thank Guido Imbens, Pat Kline, Tony Strittmatter
and seminar participants at University College London, UC Berkeley
and University of St. Gallen for helpful discussion. Financial support
from NSF grant SES \#1357499 is gratefully acknowledged. The initial
draft of this paper was prepared in October of 2016. All the usual
disclaimers apply.}}
\par\end{center}
\begin{abstract}
Let $Y$ be an outcome of interest, $X$ a \emph{vector} of treatment
measures, and $W$ a vector of pre-treatment control variables. Here
$X$ may include (combinations of) continuous, discrete, and/or non-mutually
exclusive ``treatments''. Consider the linear regression of $Y$
onto $X$ in a subpopulation homogenous in $W=w$ (formally a conditional
linear predictor). Let $b_{0}\left(w\right)$ be the coefficient vector
on $X$ in this regression. We introduce a semiparametrically efficient
estimate of the average $\beta_{0}=\mathbb{E}\left[b_{0}\left(W\right)\right]$.
When $X$ is binary-valued (multi-valued) our procedure recovers the
(a vector of) average treatment effect(s). When $X$ is continuously-valued,
or consists of multiple non-exclusive treatments, our estimand coincides
with the average partial effect (APE) of $X$ on $Y$ when the underlying
potential response function is linear in $X$, but otherwise heterogenous
across agents. When the potential response function takes a general
nonlinear/heterogenous form, and $X$ is continuously-valued, our
procedure recovers a weighted average of the gradient of this response
across individuals and values of $X$. We provide a simple, and semiparametrically
efficient, method of covariate adjustment for settings with complicated
treatment regimes. Our method generalizes familiar methods of covariate
adjustment used for program evaluation as well as methods of semiparametric
regression (e.g., the partially linear regression model).\smallskip{}

\uline{JEL Codes:} C14, C21, C31 

\uline{Keywords:}\textbf{ }Conditional Linear Predictor, Causal
Inference, Average Treatment Effect, Propensity Score, Semiparametric
Efficiency, Semiparametric Regression
\end{abstract}
\end{singlespacing}

\thispagestyle{empty} 

\pagebreak{}

\setcounter{page}{1}

\sloppy

Let $Y$ be a scalar-valued outcome of interest, $X$ a $K\times1$
vector of policy variables, and $W$ a $J\times1$ vector of additional
controls. For example $Y$ might equal hours worked, $X$ include
the real wage rate \emph{and} total unearned income ($K=2$), and
$W$ be a vector of demographic measures capturing heterogeneity in
preferences for work \citep[e.g.,][Section 4]{Pencavel_HLE86}. The
goal is to summarize how $Y$ \textendash{} labor supply \textendash{}
covaries with $X$ \textendash{} the wage rate and unearned income
\textendash{} ``holding the controls $W$ fixed''. In a second example,
$Y$ might be an end-of-year student mathematics achievement measure,
$X$ a vector containing (i) number of days absent from school, (ii)
class size and (iii) an indicator for whether the student received
supplemental tutoring. Here the vector $W$ might include beginning
of school year joint predictors of $Y$ and $X$ (e.g., prior mathematics
achievement, socioeconomic background, health indicators, and known
determinants of class size and tutoring assignment used by the school).
The goal is to summarize how math achievement covaries with attendance,
class size and supplemental tutoring conditional on $W$ \citep[cf.,][]{Gottfried_Kirksey_ER17}.

Following the prototype established by \citet{Yule_JRSS1899} over
one hundred years ago, social scientists typically report the coefficient
on $X$ in the (long) least squares fit of $Y$ onto a constant, $X$,
\emph{and} $W$ for this purpose. 

When $X$ is a scalar binary variable, the econometrician can choose
from \textendash{} in addition to least squares \textendash{} an ever
more elaborate menu of covariate adjustment methods (see \citet{Imbens_Rubin_CIBook15}
for a recent textbook introduction). Many of these methods extend
naturally to settings where $X$ is multi-valued \citep[e.g., ][]{Cattaneo_JOE10}.

When $X$ is continuously-valued, and/or consists of multiple distinct
policy variables ($K\geq2$), options are fewer \citep[cf.,][Chapter 21.6.3]{Wooldridge_EACSPDBook10}.
The partially linear regression (PLM) model
\begin{equation}
Y=X'\beta_{\mathrm{0}}+h_{0}\left(W\right)+U,\thinspace\thinspace\mathbb{E}\left[\left.U\right|W,X\right]=0,\label{eq: plm}
\end{equation}
represents one semiparametric generalization of (long) linear regression.
\citet{Chamberlain_WP86}, in an influential but never published paper,
introduced an estimator for $\beta_{0}$ in (\ref{eq: plm}) \citep[cf., ][]{Robinson_EM88}.
In later work he characterized its semiparametric efficiency bound
(SEB) \citep{Chamberlain_EM92}. 

Partially linear regression is widely, albeit heuristically, used
in empirical work. Typically researchers proceed by (i) choosing $W$
to be a rich vector of basis functions in the underlying controls
(e.g., a vector of polynomial or piecewise polynomial terms) and then
(ii) estimate $\beta_{0}$ by least squares. With discretely-valued
control variables a saturated specification for $h_{0}\left(W\right)$
is possible, at least when utilizing a very large dataset \citep[e.g., ][Section 2.3.1]{Angrist_Krueger_HLE99}.
A principled variant of this general approach is embodied in the E-Estimation
algorithm of \citet{Newey_JAE90} and \citet{Robins_Mark_Newey_BM92}. 

In this paper we propose a different approach to covariate adjustment.
Consider a subpopulation homogenous in $W=w$. Within this subpopulation
we compute the linear regression of $Y$ onto a constant and $X$
(formally a conditional linear predictor as in \citet{Wooldridge_JE99}).
Let $b_{0}\left(w\right)$ be the coefficient on $X$ in the conditional
linear regression for the subpopulation homogenous in $W=w$. We propose
a method for identifying and efficiently estimating the average regression
coefficient
\begin{equation}
\beta_{0}=\mathbb{E}\left[b_{0}\left(W\right)\right].\label{eq: average_coef}
\end{equation}
The average is over the marginal distribution of controls, $W$. 

In the absence of controls, the relationship between the linear predictor
slope coefficient and the gradient of the (possibly nonlinear) conditional
expectation function (CEF) of $Y$ given $X=x$ is well-understood
\citep[e.g., ][]{Goldberger_ACE91,Yitzhaki_JBES96}. In the presence
of controls, this relationship is rather more complicated \citep[cf.,][]{Angrist_EM98,Sloczynski_WP18}.
Our focus on averages of conditional linear predictor coefficients
allows for conditioning on $W$, while also preserving the interpretative
transparency of unconditional linear analyses. That is, $\beta_{\mathrm{0}}$,
as we demonstrate below, is easy to interpret.

When $X$ is binary-valued (multi-valued) $\beta_{0}$ coincides with
the (a vector of) average treatment effect(s); estimands familiar
from the program evaluation literature \citep[e.g.,][]{Hahn_EM98,Imbens_BM00}.
These estimands have causal interpretations under certain conditions.
Modestly extending the analysis of \citet{Wooldridge_CWP04}, we show
that this causal interpretation generalizes under a (i) heterogenous
random coefficients potential outcome structure and (ii) an unconfoundedness-type
assumption. These assumptions coincide with their program evaluation
counterparts when $X$ is binary- or multi-valued. Our semiparametric
model includes both the program evaluation model and the partially
linear regression model as special cases.

Our work is also connected to the varying coefficient model of \citet{Hastie_Tibshirani_JRSS93}.
\citet{Hastie_Tibshirani_JRSS93} focus on pointwise estimation of
$b_{0}\left(w\right)$, while we focus on (efficient) estimation of
the average $\beta_{0}=\mathbb{E}\left[b_{0}\left(W\right)\right]$. 

The relationship of our work with that of \citet{Wooldridge_CWP04}
is as follows.\footnote{The \citet{Wooldridge_CWP04} paper remains unpublished, but a textbook
treatment of the material in it can be found in Chapter 21.6.3 of
\citet{Wooldridge_EACSPDBook10}.} We both study the same functional of the joint distribution of $W$,
$X$ and $Y$ (see Equation (\ref{eq: generalized_ipw}) below). Relative
to \citet{Wooldridge_CWP04} we provide an average partial effect
interpretation of this estimand under (i) weaker assumptions when
maintaining a correlated random coefficient potential outcome structure
and (ii) a new weighted average partial effect interpretation under
a general potential response function structure. These are useful,
but relatively modest generalizations. More significantly we (i) provide
distribution theory for the estimator proposed by \citet{Wooldridge_CWP04},
(ii) characterize the semiparametric efficiency bound (SEB) for $\beta_{0}$,
and (iii) introduce a new locally efficient estimator. The procedure
proposed by \citet{Wooldridge_CWP04} is inefficient. 

Another feature of our estimator is computational simplicity. Let
$\hat{\mu}_{W}=\frac{1}{N}\sum_{i=1}^{N}W_{i}$ be the sample mean
of $W$. A common approach to modeling heterogeneous effects in applied
work is to compute the least squares fit of $Y$ onto a constant,
$W-\hat{\mu}_{W}$, $\left(W-\hat{\mu}_{W}\right)\otimes X$, and
$X$. As is well-known from textbook treatments on interaction terms
in linear regression analysis, centering the control variable vector,
$W$, about is mean in this way ensures that the coefficient on $X$
captures an average effect. This approach essentially coincides with
Oaxaca-Blinder type methods of covariate adjustment popular in labor
economics \citep[e.g., ][]{Kline_EL14}. One variant of our procedure
involves computing the exact same regression, but where $X$ is instead
instrumented with a particular function of its conditional distribution
given $W$ (i.e., of the ``generalized'' propensity score). Theorems
\ref{thm: Large-Sample} and \ref{thm:DoublyRobust} below show that
this small modification to a familiar estimation procedure delivers
considerable gains.

The next section introduces our average linear regression model. We
provide a statistical definition of $\beta_{\mathrm{0}}$ as well
as sets of assumptions under which it has a causal \textendash{} average
partial effect (APE) \textendash{} interpretation. Section \ref{sec: seb}
presents the semiparametric efficiency bound for $\beta_{0}$. Section
\ref{sec: estimation} studies the large sample properties of the
\citet{Wooldridge_CWP04} estimator. We also introduce our new estimator
and present its large sample properties. Finally, in Section \ref{sec: examples},
we connect our results with prior work on efficient estimation of
average treatments effects as well as the partially linear semiparametric
regression model. We end our paper with a small simulation study in
Section \ref{sec: Monte_Carlo}. All proofs are collected in the Appendix
or the supplemental materials.

\section{\label{sec: model}Average linear regression model}

We begin with a conventional sampling assumption.
\begin{assumption}
\textsc{(Random Sampling)}\label{ass: random_sampling} Let $\left\{ \left(W_{i}',X_{i}',Y_{i}\right)'\right\} _{i=1}^{\infty}$
be a sequence of independent and identically distributed random draws
from some population $F_{W,X,Y}$ with $\mathbb{E}\left[\left.Y^{2}\right|W=w\right]<\infty$
and $\mathbb{E}\left[\left.\left\Vert X\right\Vert ^{2}\right|W=w\right]<\infty$
for all $w\in\mathbb{W}.$
\end{assumption}
The finite moment restrictions included in Assumption \ref{ass: random_sampling}
ensure that a conditional linear predictor (CLP) is well-defined for
all $w\in\mathbb{W}$. 

Let 
\begin{equation}
e_{0}\left(w\right)=\mathbb{E}\left[\left.X\right|W=w\right]\label{eq: e(W)}
\end{equation}
 be the conditional mean of $X$ given $W=w$ and
\begin{equation}
v_{0}\left(w\right)=\mathbb{V}\left(\left.X\right|W=w\right)\label{eq: v(W)}
\end{equation}
the corresponding conditional variance. We also require that $X$
vary conditional on $W=w$.
\begin{assumption}
\textsc{(Overlap)}\label{ass: overlap} For all $w\in\mathbb{W}$
and any non-zero column vector $t$, $t'v_{0}\left(w\right)t\geq\kappa>0.$
\end{assumption}
Assumption \ref{ass: overlap} ensures that the CLP is uniquely defined.
In the absence of conditioning it is equivalent to linear independence
of the elements of $X$. When $X$ is binary $v_{0}\left(W\right)=e_{0}\left(W\right)\left(1-e_{0}\left(W\right)\right)$
with $e_{0}\left(W\right)=\Pr\left(\left.X=1\right|W\right)$ equal
to the propensity score; in this case Assumption \ref{ass: overlap}
coincides with the familiar strong overlap assumption from the program
evaluation literature. More generally Assumption \ref{ass: overlap}
implies that $X$ varies conditional on $W=w$ for \emph{all} $w\in\mathbb{W}.$

Under Assumptions \ref{ass: random_sampling} and \ref{ass: overlap}
the conditional linear predictor is well-defined for all $w\in\mathbb{W}$.
\citet[Section 4]{Wooldridge_JE99} provides a self-contained introduction
to conditional linear predictors. The following definition and lemma
is taken from \citet{Wooldridge_JE99}.
\begin{defn}
\textsc{(Conditional Linear Predictor)}\label{def: clp} The mean
squared error minimizing linear predictor of $Y$ given $X$ \emph{conditional}
on $W=w$, henceforth the conditional linear predictor (CLP), equals
\begin{equation}
\mathbb{E}^{*}\left[\left.Y\right|X;W=w\right]\overset{def}{\equiv}a_{0}\left(w\right)+X'b_{0}\left(w\right)\label{eq: clp}
\end{equation}
with
\begin{align}
a_{0}\left(w\right) & \overset{def}{\equiv}\mathbb{E}\left[\left.Y\right|W=w\right]-e_{0}\left(w\right)'b_{0}\left(w\right)\label{eq: clp_coef}\\
b_{0}\left(w\right) & \overset{def}{\equiv}v_{0}\left(w\right)^{-1}\mathbb{C}\left(\left.X,Y\right|W=w\right).\nonumber 
\end{align}
It is straightforward to show that the prediction error $U=Y-\mathbb{E}^{*}\left[\left.Y\right|X;W\right]$
is conditionally mean zero and conditionally uncorrelated with $X$.
This property of $\mathbb{E}^{*}\left[\left.Y\right|X;W\right]$ will
prove useful for what follows.
\end{defn}
\begin{lem}
\label{lem: Wooldridge}\citet[Lemma 4.1]{Wooldridge_JE99}. Let $U\overset{def}{\equiv}Y-a_{0}\left(W\right)-X'b_{0}\left(W\right)$,
then $\mathbb{E}\left[\left.U\right|W=w\right]=0$ and $\mathbb{E}\left[\left.XU\right|W=w\right]=0$
for all $w\in\mathbb{W}.$
\end{lem}

\subsection*{Identification of the average regression slope}

We begin by presenting a convenient representation of the average
slope coefficient $\beta_{\mathrm{0}}=\mathbb{E}\left[b_{0}\left(W\right)\right]$
in terms of the joint distribution of $\left(W',X',Y\right)'$. The
most direct representation follows directly from (\ref{eq: clp_coef}):
\[
\beta_{\mathrm{0}}=\mathbb{E}\left[v_{0}\left(W\right)^{-1}\mathbb{C}\left(\left.X,Y\right|W\right)\right].
\]
For our purposes, however, an alternative representation of $\beta_{\mathrm{0}}$
is more convenient; both for our semiparametric efficiency bound (SEB)
analysis and for the approach to estimation developed below. Using
the law of iterated expectations and the definition of conditional
covariance we get, under Assumptions \ref{ass: random_sampling} and
\ref{ass: overlap},
\begin{align*}
\mathbb{E}\left[v_{0}\left(W\right)^{-1}\left(X-e_{0}\left(W\right)\right)Y\right] & =\mathbb{E}\left[v_{0}\left(W\right)^{-1}\mathbb{E}\left[\left.\left(X-e_{0}\left(W\right)\right)Y\right|W\right]\right]\\
 & =\mathbb{E}\left[v_{0}\left(W\right)^{-1}\mathbb{C}\left(\left.X,Y\right|W\right)\right].
\end{align*}

Applying definition (\ref{eq: clp_coef}) then gives our preferred
estimand representation:
\begin{align}
\beta_{0} & =\mathbb{E}\left[v_{0}\left(W\right)^{-1}\left(X-e_{0}\left(W\right)\right)Y\right].\label{eq: generalized_ipw}
\end{align}
\citet{Wooldridge_CWP04} emphasizes the coincidence between (\ref{eq: generalized_ipw})
and the average partial effect of $X$ on $Y$ associated with a particular
correlated random coefficients (CRC) potential outcomes structure.
This endows $\beta_{0}$ with causal meaning. While we also develop
this connection below, we wish to initially emphasize that (\ref{eq: generalized_ipw})
is also just one way of representing a population average of conditional
linear predictor coefficients. Under Assumptions \ref{ass: random_sampling}
and \ref{ass: overlap} the expectation in (\ref{eq: generalized_ipw})
is well-defined and $\beta_{0}$ is simply a ``statistical'' estimand.
We are interested in estimating it as precisely as possible.

\subsection*{Causal interpretation}

In this subsection we show that (\ref{eq: generalized_ipw}) admits
a causal interpretation under a particular treatment response model
and selection on observables type assumption. As noted earlier, this
interpretation was previously emphasized by \citet{Wooldridge_CWP04},
but under stronger conditions than we maintain here. 

Associated with each agent in the target population is an individual-specific
potential response function, $Y\left(x\right)$, which maps counterfactual
values of the input vector $X$ into their corresponding (potential)
outcomes. The observed outcome coincides with the value of the potential
response function at the observed input level $X$: $Y=Y\left(X\right)$.
We assume that $Y\left(x\right)$ is linear in $x$, but otherwise
heterogeneous across individuals: 
\begin{equation}
Y\left(x\right)=A+x'B,\label{eq: CRC_potential_response}
\end{equation}
where $A$ and $B$ are an individual-specific intercept and slope
vector respectively. 

Equation (\ref{eq: CRC_potential_response}) allows for each individual
to have their own potential response function, but restricts them
to be linear in $X$. When $X$ is binary, or multi-valued, linearity
is unrestrictive. For example, in the binary case, we have the potential
outcome under control ($X=0$) and active $(X=1$) treatment equal
to $Y\left(0\right)=A$ and $Y\left(1\right)=A+B$. In the multi-valued
treatment setting of \citet{Imbens_BM00} and \citet{Cattaneo_JOE10},
with $X$ a vector of treatment indicators for $K$ mutually exclusive
treatments, we have $Y\left(0\right)=A$ and $Y\left(k\right)=A+B_{k}$
for $k=1,\ldots,K$. In contrast, when $X$ is ordered, continuously
valued, or includes multiple treatments/policies, linearity is restrictive.

Consider the following thought experiment: draw a unit at random and
(exogenously) increase the value of the $k^{th}$ component of $X$
by one unit. The expected effect of this intervention is $\mathbb{E}\left[B_{k}\right]$.
In the binary- and multi-valued treatment setting $\mathbb{E}\left[B_{k}\right]$
corresponds to an average treatment effect (ATE)
\[
\mathbb{E}\left[B_{k}\right]=\mathbb{E}\left[Y\left(k\right)-Y\left(0\right)\right].
\]
More generally $\mathbb{E}\left[B_{k}\right]$ equals the \emph{average
partial effect} (APE) of a unit increase in $X_{k}$. This estimand
was introduced in a panel data setting by \citet{Chamberlain_HBE84};
general expositions, with additional results, are available in \citet{Blundel_Powell_WC03}
and \citet{Wooldridge_IIEM05}.

Under the following assumption, in addition to those introduced above,
we can show that $\beta_{0}$ coincides with the APE vector, $\mathbb{E}\left[B\right]$.
\begin{assumption}
\textsc{\label{ass: CondExog}(Conditional Exogeneity) }\emph{For
all $w\in\mathbb{W}$ and }\textup{\emph{$k,l=1,\ldots,K$, and under
potential responses of the form given in (\ref{eq: CRC_potential_response})}}
\begin{equation}
\mathbb{C}\left(\left.A,X_{k}\right|W=w\right)=\mathbb{C}\left(\left.B,X_{k}\right|W=w\right)=\mathbb{C}\left(\left.B,X_{k}X_{l}\right|W=w\right)=0.\label{eq: exogeneity}
\end{equation}
\end{assumption}
Assumption \ref{ass: CondExog} restricts the form of any dependence
between the potential response function, $Y\left(x\right)=A+x'B$,
and the treatment vector actually chosen by the respondent, $X$.
It is a conditional exogeneity or selection on observables type assumption.
To see this observe that when $X$ is binary Assumption \ref{ass: CondExog}
coincides with the standard mean independence assumption familiar
from the program evaluation literature, implying that 
\[
\mathbb{E}[Y(x)|X,W]=\mathbb{E}[Y(x)|W].
\]
In the multi-valued treatment setting Assumption \ref{ass: CondExog}
also coincides with standard generalizations of the mean independence
assumption \citep[cf., ][]{Imbens_BM00}. See also Section \ref{sec: examples}
below.

When the linearity of (\ref{eq: CRC_potential_response}) is restrictive,
as occurs when $X$ includes continuously-valued components, or non-mutually
exclusive binary inputs, Assumption \ref{ass: CondExog} is less restrictive
than other possible formulations of conditional exogeneity. For example,
\citet{Wooldridge_CWP04,Wooldridge_EACSPDBook10} works with the identifying
restrictions
\begin{equation}
\mathbb{E}\left[\left.X\right|W,A,B\right]=\mathbb{E}\left[\left.X\right|W\right]=e_{0}(W),\thinspace\thinspace\mathbb{V}\left(\left.X\right|W,A,B\right)=\mathbb{V}\left(\left.X\right|W\right)=v_{0}(W)\label{eq: exogeneity_alt1}
\end{equation}
which imply (\ref{eq: exogeneity}), but are generally stronger. An
even stronger notion of conditional exogeneity is
\begin{equation}
\mathbb{E}\left[\left.A\right|X,W\right]=\mathbb{E}\left[\left.A\right|W\right],\thinspace\thinspace\mathbb{E}\left[\left.B\right|X,W\right]=\mathbb{E}\left[\left.B\right|W\right].\label{eq: exogeneity_alt2}
\end{equation}
Assumption \ref{ass: CondExog} is (apparently) the weakest assumption
necessary to equate $\beta_{0}$ with the average partial effect of
$X$ on $Y$ when the potential response function takes form (\ref{eq: CRC_potential_response}).
The estimator we introduce below will remain consistent under the
stronger restrictions, (\ref{eq: exogeneity_alt1}) and (\ref{eq: exogeneity_alt2}),
but will generally not be semiparametrically efficient in those cases.
We elaborate further on this observation below.
\begin{prop}
\textsc{(Average Partial Effect Identification) }Under Assumptions
\ref{ass: random_sampling}, \ref{ass: overlap} and \ref{ass: CondExog}
the average of the CLP coefficients, $\beta_{0}=\mathbb{E}\left[b_{0}\left(W\right)\right]$,
and the average partial effect (APE), $\mathbb{E}\left[B\right]$,
coincide:
\[
\beta_{0}=\mathbb{E}\left[B\right].
\]
\end{prop}
\begin{proof}
\citet{Wooldridge_CWP04} demonstrates the equality under the stronger
restriction (\ref{eq: exogeneity_alt1}). Under Assumption \ref{ass: CondExog},
however, the proof proceeds differently. Given the linear potential
response (\ref{eq: CRC_potential_response}) and by lemma (\ref{lem: Wooldridge}),
we have the $1+K$ conditional moment restrictions 
\begin{eqnarray}
\mathbb{E}\left[\left.U\right|W=w\right] & = & \mathbb{E}\left[\left.A-a_{0}\left(W\right)\right|w\right]+\mathbb{E}\left[\left.X'\left(B-b_{0}\left(W\right)\right)\right|w\right]=0\nonumber \\
\mathbb{E}\left[\left.XU\right|W=w\right] & = & \mathbb{E}\left[\left.X\left(A-a_{0}\left(W\right)\right)\right|w\right]+\mathbb{E}\left[\left.XX'\left(B-b_{0}\left(W\right)\right)\right|w\right]=0.\label{eq: conditional_moments}
\end{eqnarray}
Under Assumption \ref{ass: CondExog} conditions (\ref{eq: conditional_moments})
simplify to
\begin{align*}
\left\{ \mathbb{E}\left[\left.A\right|w\right]-a_{0}\left(w\right)\right\} +e_{0}\left(w\right)'\left\{ \mathbb{E}\left[\left.B\right|w\right]-b_{0}\left(w\right)\right\}  & =0\\
e_{0}\left(w\right)\left\{ \mathbb{E}\left[\left.A\right|w\right]-a_{0}\left(w\right)\right\} +\mathbb{E}\left[\left.XX'\right|w\right]\left\{ \mathbb{E}\left[\left.B\right|w\right]-b_{0}\left(w\right)\right\}  & =0
\end{align*}
or, in matrix form,
\[
\left[\begin{array}{cc}
1 & e_{0}\left(w\right)'\\
e_{0}\left(w\right) & \mathbb{E}\left[\left.XX'\right|w\right]
\end{array}\right]\left(\begin{array}{c}
\mathbb{E}\left[\left.A\right|w\right]-a_{0}\left(w\right)\\
\mathbb{E}\left[\left.B\right|w\right]-b_{0}\left(w\right)
\end{array}\right)=\left(\begin{array}{c}
0\\
0
\end{array}\right).
\]
Under the Assumption \ref{ass: overlap} the first matrix to the left
of the equality is invertible for all $w\in\mathbb{W}$. This implies
that $\mathbb{E}\left[\left.A\right|W=w\right]=a_{0}\left(w\right)$
and $\mathbb{E}\left[\left.B\right|W=w\right]=b_{0}\left(w\right)$
for all $w\in\mathbb{W}$. The result follows by iterated expectations.
\end{proof}

\subsection*{Causal interpretation under misspecification}

\citet[Section 2.3.1]{Angrist_Krueger_HLE99} and \citet[Chapter 3.3]{Angrist_Pischke_MHE09}
emphasize that when $X$ is a continuously-valued random variable
its slope coefficient in the linear predictor of $Y$ onto a constant,
$X$ and the vector of ``saturated'' controls admits a weighted
average derivative interpretation when the potential response function
takes a general nonlinear form \citep[cf.,][Lemma 5]{Angrist_Graddy_Imbens_ReStud00}.
Angrist and Krueger's \citeyearpar{Angrist_Krueger_HLE99} expression
is also isomorphic to the probability limit of the E-Estimator of
\citet{Newey_JAE90} and \citet{Robins_Mark_Newey_BM92} 
\begin{equation}
\beta_{\mathrm{E}}=\frac{\mathbb{E}\left[Y\left(X-e_{0}\left(W\right)\right)\right]}{\mathbb{E}\left[X\left(X-e_{0}\left(W\right)\right)\right]}\label{eq: E_estimand}
\end{equation}
when the partially linear regression structure, equation (\ref{eq: plm})
above, is incorrect. 

In this section, using similar arguments to those appearing in \citet[Lemma 5]{Angrist_Graddy_Imbens_ReStud00}
and \citet[Lemma A.1]{Graham_Imbens_Ridder_NBER10}, we provide a
representation result for $\beta_{0}$ under a general potential response
function.

Assume that the potential response function is nonlinear and heterogeneous
such that $Y\left(x\right)=h\left(x,U\right)$. Further assume, stronger
than Assumption \ref{ass: CondExog} above, that $U$ is conditionally
independent of $X$ given $W=w$ for all $w\in\mathbb{W}$. \citet{Blundel_Powell_WC03}
show that the \emph{partial mean} $\mathbb{E}_{W}\left[\mathbb{E}\left[\left.Y\right|W,X=x\right]\right]$
identifies the average structural function (ASF) $m\left(x\right)=\mathbb{E}_{U}\left[h\left(x,U\right)\right]$
when the support of $W$ given $X=x$ coincides with its marginal
support. \citet{Newey_ET94b} provides an explicit partial mean estimator
and derives in asymptotic properties. 

Here we show that our average regression slope estimand, $\beta_{0}$,
can be expressed as a weighted average of the gradient of $h\left(X,U\right)$.
This provides a causal interpretation of $\beta_{0}$ under a general
potential response function. To present this result we replace Assumption
\ref{ass: CondExog} with:
\begin{assumption}
\textsc{(Nonlinear Potential Response Function)}\label{ass: wgt_der}\textsc{
}(i) $X$ is a continuous scalar random variable with bounded support
$\mathbb{X}=\left[\underline{x},\overline{x}\right]$, (ii) the conditional
density function of $X$ given $W=w$ is bounded and bounded away
from zero for all $\left(w,x\right)\in\mathbb{W}\times\mathbb{X}$,
(iii) $Y=h\left(X,U\right)$ with $h\left(x,u\right)$ a continuously
differentiable function of $x$ for all $\left(x,u\right)\in\mathbb{X}\times\mathbb{U}$
and $\underline{h}\left(u\right)=h\left(\underline{x},u\right)$ finite
for all $u\in\mathbb{U}$, and (iv) $U$ is conditionally independent
of $X$ given $W=w$ for all $w\in\mathbb{W}$.
\end{assumption}
\begin{prop}
\label{prop: wgt_der}(\textsc{Weighted Average Derivative Representation})
Under Assumptions \ref{ass: random_sampling}, \ref{ass: overlap}
and \ref{ass: wgt_der}
\[
\beta_{0}=\mathbb{E}\left[\omega\left(W,X\right)\frac{\partial h\left(X,U\right)}{\partial x}\right]
\]
where
\[
\omega\left(w,x\right)=\frac{1}{f_{\left.X\right|W}\left(\left.x\right|w\right)}\frac{\mathbb{E}\left[\left.X-e_{0}\left(W\right)\right|W=w,X\geq x\right]\left(1-F_{\left.X\right|W}\left(\left.x\right|w\right)\right)}{\int_{\underline{x}}^{\bar{x}}\mathbb{E}\left[\left.X-e_{0}\left(W\right)\right|W=w,X\geq v\right]\left(1-F_{\left.X\right|W}\left(\left.v\right|w\right)\right)\mathrm{d}v}.
\]
\end{prop}
\begin{proof}
See the Supplemental Web Appendix.
\end{proof}
A key feature of the weighting function $\omega\left(w,x\right)$
is that its \emph{conditional} mean, $\mathbb{E}\left[\left.\omega\left(W,X\right)\right|W=w\right]$,
equals $1$ for \emph{every} value of $w\in\mathbb{W}$. Furthermore,
Lemma A.1 of \citet{Graham_Imbens_Ridder_NBER10} implies that, conditional
on $W=w$, the weight given to $\frac{\partial h\left(X,U\right)}{\partial x}$
is highest for those values of $X$ near its conditional mean, $\mathbb{E}\left[\left.X\right|W=w\right]$,
and lowest for those at the boundary of its support, $\underline{x}$
and $\overline{x}$. 

These features of the weights appearing in Proposition \ref{prop: wgt_der}
imply the following intuitive interpretation: (i) for each value of
$w\in\mathbb{W}$ compute a weighted average of $\frac{\partial h\left(X,U\right)}{\partial x}$,
where the average emphasizes values of $X$ near its conditional mean
given $W=w$, (ii) average these (weighted average) gradients over
the marginal distribution of $W$. This indicates that $\beta_{0}$
only differs from the unweighted average $\mathbb{E}\left[\frac{\partial h\left(X,U\right)}{\partial x}\right]$
due to variation in $\omega\left(W,X\right)$ \emph{within} $W=w$
cells. The contribution of each subpopulation, defined in terms of
the control, $W$, mirrors its density in the sampled population.
Since $W$ proxies for $U$ in this set-up we \emph{are }averaging
over the correct heterogeneity distribution. 

More precisely, since $\mathbb{E}\left[\left.\omega\left(W,X\right)\right|W=w\right]=1$,
we have that, using the definition of conditional covariance,
\begin{equation}
\beta_{0}-\mathbb{E}\left[\frac{\partial h\left(X,U\right)}{\partial x}\right]=\mathbb{E}\left[\mathbb{C}\left(\left.\omega\left(W,X\right),\frac{\partial h\left(X,U\right)}{\partial x}\right|W\right)\right].\label{eq: bias_beta0}
\end{equation}
The bias of $\beta_{0}$ for $\mathbb{E}\left[\frac{\partial h\left(x,U\right)}{\partial x}\right]$
is therefore solely due to \emph{conditional} covariance between the
weight function and the gradient of interest within subpopulations
homogenous in $W$.

In contrast to the one for $\beta_{0}$, the weight function appearing
in the weighted average derivative representation result of \citet{Angrist_Krueger_HLE99}
or \citet{Angrist_Pischke_MHE09} for $\beta_{\mathrm{E}}$ is only
\emph{unconditionally} mean zero. This implies that $\beta_{\mathrm{E}}$
averages over the incorrect heterogeneity distribution \emph{as well
as} the incorrect policy variable distribution. 
\begin{align}
\beta_{\mathrm{E}}-\mathbb{E}\left[\frac{\partial h\left(X,U\right)}{\partial x}\right]= & \mathbb{E}\left[\mathbb{C}\left(\left.\omega\left(W,X\right),\frac{\partial h\left(X,U\right)}{\partial x}\right|W\right)\right]\label{eq: bias_bE}\\
 & +\mathbb{C}\left(\mathbb{E}\left[\left.\omega\left(W,X\right)\right|W\right],\mathbb{E}\left[\left.\frac{\partial h\left(X,U\right)}{\partial x}\right|W\right]\right).\nonumber 
\end{align}
If the ultimate object of interest is the average derivative $\mathbb{E}\left[\frac{\partial h\left(X,U\right)}{\partial x}\right]$,
then, relative to (\ref{eq: bias_bE}), a focus on $\beta_{0}$ eliminates
one source of potential bias. Namely that the weight function may
over- or under-emphasize various subpopulations defined in terms of
their value of the control variable vector $W$. In this case $\mathbb{E}\left[\left.\omega\left(W,X\right)\right|W\right]$
may not equal one and the second term to the right of the equality
in (\ref{eq: bias_bE}) may be non-zero.\footnote{To be clear $\omega\left(W,X\right)$ are different functions in expressions
(\ref{eq: bias_beta0}) and (\ref{eq: bias_bE}); for its form in
the latter case see \citet{Angrist_Krueger_HLE99} or \citet{Angrist_Pischke_MHE09}.}

\subsection*{Motivating $\beta_{0}$}

Our focus on averages of conditional linear predictor slope coefficients
is motivated by a combination of principled and pragmatic reasons. 

First, the kitchen sink long regression remains a workhorse of everyday
empirical social science research. Our model extends kitchen sink
regression in an easy to understand way. Relative to the partially
linear regression model, our model allows for heterogenous responses
of $Y$ to variation in $X$; a feature likely to be both empirically
relevant and \emph{a priori} attractive to researchers. 

Second, $\beta_{0}$ has a causal interpretation under additional
assumptions. When the potential response function is linear, but heterogeneous
across agents, it coincides with an average partial effect (APE) under
a selection on observables type assumption. When $X$ is binary- or
multi-valued, as in the program evaluation literature, it coincides
with the well-known average treatment effect (ATE). Our causal model
nests the usual one as a special case, but accommodates continuous
and/or multiple treatments as well (albeit under restrictions). 

Third, in the presence of misspecification $\beta_{0}$ coincides
with a weighted average of the derivative of a general non-linear
potential response function. This weighted average derivative is more
interpretable than existing representation results; for example those
of \citet{Angrist_Krueger_HLE99} for $\beta_{\mathrm{E}}$. 

Fourth, as we show next, $\beta_{0}$ is $\sqrt{N}$ estimable (or
regularly identified). This is not the case for, say, a partial mean
with a continuous policy variable \citep[e.g., ][]{Newey_ET94b}.
Regular identification suggests that estimation is practically feasible
and we present one such feasible estimator below.

Ultimately the balance between ease of interpretation under various
population assumptions and, as we show below, ease of estimation,
provides the strongest case for focusing on $\beta_{0}$.

\section{Semiparametric efficiency bound\label{sec: seb}}

Using the method of calculation outlined by \citet{Bickel_et_al_Bk93}
and \citet{Newey_JAE90}, we derive the semiparametric variance bound
for $\beta_{0}$ of,

\begin{equation}
\mathcal{I}(\beta_{0})^{-1}=\mathbb{E}\left[\Omega_{0}(W)\right]+\mathbb{V}(b_{0}(W)),\label{eq: Inverse_Information_CRC}
\end{equation}
where 
\[
\Omega_{0}(w)=\mathbb{E}\left[\left.v_{0}(W)^{-1}\left(X-e_{0}\left(W\right)\right)UU'\left\{ v_{0}(W)^{-1}\left(X-e_{0}\left(W\right)\right)\right\} '\right|W=w\right].
\]

The corresponding efficient influence function equals
\begin{align}
\mathrm{\psi}_{\beta}^{\mathrm{eff}}\left(Z,\beta_{0},g_{0}\left(W\right),h_{0}\left(W\right)\right)= & v_{0}\left(W\right)^{-1}\left(X-e_{0}\left(W\right)\right)\left(Y-a_{0}\left(W\right)-X'b_{0}\left(W\right)\right)\label{eq: EfficientInfluenceFunction_CRC}\\
 & +\left(b_{0}\left(W\right)-\beta_{\mathrm{0}}\right)\nonumber 
\end{align}
with $Z=\left(W',X',Y\right)'$, $g\left(W\right)=\left(e(W),v(W)\right)$
and $h\left(W\right)=\left(a\left(W\right),b\left(W\right)\right)$.
\begin{thm}
\label{thm: SEB}\textsc{(Semiparametric Efficiency Bound)} The efficient
influence function for $\beta_{\mathrm{0}}=\mathbb{E}\left[b_{0}\left(W\right)\right]$
in the semiparametric problem established by Definition \ref{def: clp}
and Assumptions \ref{ass: random_sampling} and \ref{ass: overlap}
equals (\ref{eq: EfficientInfluenceFunction_CRC}).
\end{thm}
\begin{proof}
See Appendix \ref{appendix: proofs}.

We also have the following corollary, which is similar to a result
for the binary case due to \citet{Robins_Rotnitzky_Zhao_JASA94},
\citet{Hahn_EM98} and \citet{Chen_Hong_Tarozz_AS08}. This corollary
will be useful when we discuss locally efficient estimation in Section
(\ref{sec: estimation}).
\end{proof}
\begin{cor}
\label{cor: redundancy_of_f(x|w)}\textsc{ (Redundancy) }Let $f\left(\left.x\right|w;\phi\right)$
be a parametric family of conditional distributions for $X$ given
$W$ with $f_{0}\left(\left.x\right|w\right)=f\left(\left.x\right|w;\phi\right)$
at some unique $\phi=\phi_{0}$. The knowledge that $f_{0}\left(\left.x\right|w\right)$
is a member of the family $f\left(\left.x\right|w;\phi\right)$ does
not change the semiparametric efficiency bound for $\beta_{0}.$
\end{cor}
\begin{proof}
See the Supplemental Web Appendix.
\end{proof}
See \citet{Frolich_ER04} and \citet{Graham_Pinto_Egel_JBES16} for
additional intuition about results like Corollary \ref{cor: redundancy_of_f(x|w)}.

\subsection*{Double robustness property of the efficient influence function}

Before introducing our estimator in the next section we highlight
an important property of the efficient influence function for $\beta_{\mathrm{0}}$. 

Consider replacing $h_{0}\left(W\right)=\left(a_{0}\left(W\right),b_{0}\left(W\right)\right)$
in (\ref{eq: EfficientInfluenceFunction_CRC}) with the incorrect
conditional linear predictor coefficients $h_{*}\left(W\right)=\left(a_{*}\left(W\right),b_{*}\left(W\right)\right)$.
Use the notation $U_{*}=\left(Y-a_{*}\left(W\right)-X'b_{*}\left(W\right)\right)$
to emphasize that $U_{*}$ is the prediction error associated with
an arbitrary conditional linear predictor (which need not be the mean
squared error minimizing one). Note that $U_{*}$ will not be conditionally
mean zero or conditionally uncorrelated with $X$ (i.e., $\mathbb{E}\left[\left.U_{*}\right|W\right]\neq0$
and $\mathbb{E}\left[\left.XU_{*}\right|W\right]\neq0$). Nevertheless,
as long as $e_{0}\left(X\right)$ and $v_{0}\left(W\right)$ equal
the true conditional mean and variance of $X$ given $W$, we have
the pair of equalities, using iterated expectations,
\begin{align*}
\mathbb{E}\left[v_{0}\left(W\right)^{-1}\left(X-e_{0}\left(W\right)\right)a_{*}\left(W\right)\right]= & 0\\
\mathbb{E}\left[v_{0}\left(W\right)^{-1}\left(X-e_{0}\left(W\right)\right)X'b_{*}\left(W\right)\right]= & \mathbb{E}\left[b_{*}\left(W\right)\right]
\end{align*}
(the second equality follows from the fact that $\mathbb{E}\left[\left.\left(X-e_{0}\left(W\right)\right)X'\right|W\right]=v_{0}\left(W\right)$).

Therefore (\ref{eq: EfficientInfluenceFunction_CRC}) remains mean
zero even if the nuisance functions $h_{0}\left(W\right)=\left(a_{0}\left(W\right),b_{0}\left(W\right)\right)$
are replaced by arbitrary functions of $W$:
\begin{equation}
\mathbb{E}\left[\mathrm{\psi}_{\beta}^{\mathrm{eff}}\left(Z,\beta_{0},g_{0}\left(W\right),h_{*}\left(W\right)\right)\right]=0.\label{eq: dr_condition1}
\end{equation}
One special choice of $h_{*}(W)$ is the zero vector. This choice
directly recovers the representation of $\beta_{\mathrm{0}}$ derived
earlier (Equation (\ref{eq: generalized_ipw}) above). In moment condition
form
\[
\mathbb{E}\left[v_{0}\left(W\right)^{-1}\left(X-e_{0}\left(W\right)\right)Y-\beta_{0}\right]=0.
\]

Next consider replacing $g_{0}\left(W\right)=\left(e_{0}\left(W\right),v_{0}\left(W\right)\right)$
in (\ref{eq: EfficientInfluenceFunction_CRC}) with the incorrect
conditional mean and variance functions $g_{*}\left(W\right)=\left(e_{*}\left(W\right),v_{*}\left(W\right)\right)$.
Use the notation $U_{0}=\left(Y-a_{0}\left(W\right)-X'b_{0}\left(W\right)\right)$
to emphasize that $U_{0}$ is the prediction error associated with
the mean squared error minimizing conditional linear prediction of
$Y$ given $X$ conditional on $W$. By Lemma \ref{lem: Wooldridge}
$\mathbb{E}\left[\left.U_{0}\right|W\right]=0$ and $\mathbb{E}\left[\left.XU_{0}\right|W\right]=0$.
Therefore
\begin{align*}
\mathbb{E}\left[\mathrm{\psi}_{\beta}^{\mathrm{eff}}\left(Z,\beta_{0},g_{*}\left(W\right),h_{0}\left(W\right)\right)\right]= & \mathbb{E}\left[v_{*}\left(W\right)^{-1}\left(X-e_{*}\left(W\right)\right)\left(Y-a_{0}\left(W\right)-X'b_{0}\left(W\right)\right)\right]\\
 & +\mathbb{E}\left[\left(b_{0}\left(W\right)-\beta_{0}\right)\right]\\
= & \mathbb{E}\left[v_{*}\left(W\right)^{-1}\mathbb{E}\left[\left.XU_{0}\right|W\right]\right]-\mathbb{E}\left[v_{*}\left(W\right)^{-1}e_{*}\left(W\right)\mathbb{E}\left[\left.U_{0}\right|W\right]\right]\\
= & 0.
\end{align*}
Hence (\ref{eq: EfficientInfluenceFunction_CRC}) also remains mean
zero even if the nuisance functions $g_{0}\left(W\right)=\left(e_{0}\left(W\right),v_{0}\left(W\right)\right)$
are replaced by arbitrary functions of $W$.

Moment (\ref{eq: EfficientInfluenceFunction_CRC}) has the so-called
doubly robust property of \citet{Scharfstein_Rotnitzky_Robins_JASA99b}
\citep[cf., ][]{Ruud_JOE86}. It is mean zero as long as one of the
two nuisance functions, $g\left(W\right)$ or $h\left(W\right)$,
coincides with its population one. We exploit this property when constructing
our estimator in the next section.

\section{Estimation\label{sec: estimation}}

In this section we present a locally semiparametrically efficient
estimate of $\beta_{0}$. To motivate the precise form of our estimator
we also discuss the estimator proposed by \citet{Wooldridge_CWP04}.
A textbook presentation of this estimator is available in Chapter
21.6.3 of \citet{Wooldridge_EACSPDBook10}.

For the purposes of estimation we impose a parametric restriction
on the conditional distribution of X given W. Since the distribution
of X given W is ancillary for $\beta_{0},$ this parametric restriction
does not change the semiparametric efficiency bound (cf., Corollary
\ref{cor: redundancy_of_f(x|w)}). We call, borrowing nomenclature
from related settings \citep[e.g.,][]{Hirano_Imbens_BM04}, the resulting
model for $X$ the \emph{generalized propensity score}.
\begin{assumption}
\label{ass: genalize_p_score}\textsc{(Generalized Propensity Score)}
$f\left(\left.x\right|w;\phi\right)$ is a parametric family of densities
indexed by $\phi\in\Phi\subset\mathbb{R}^{L}$ with (i) $f_{0}\left(\left.x\right|w\right)=\mbox{f\ensuremath{\left(\left.x\right|w;\phi_{0}\right)}}$
at some unique $\phi_{0}\in\mathrm{int}\left(\Phi\right)$, (ii) a
maximum likelihood estimate (MLE) of $\phi_{0}$ equal to 
\[
\hat{\phi}=\arg\underset{\phi\in\Phi}{\max}\sum_{i=1}^{N}\ln f\left(\left.X_{i}\right|W_{i};\phi\right)
\]
with a score vectors of $\mathbb{S}_{\phi}\left(\left.X\right|W;\phi\right)=\nabla_{\phi}f\left(\left.X\right|W;\phi\right)/f\left(\left.X\right|W;\phi\right)$,
(iii) $\hat{\phi}\stackrel{p}{\rightarrow}\phi_{0}$ with $\mathbb{E}\left[\mathbb{S}_{i}\mathbb{S}_{i}'\right]$
non-singular and the asymptotically linear representation 
\begin{equation}
\sqrt{N}\left(\hat{\phi}-\phi_{0}\right)=\mathbb{E}\left[\mathbb{S}_{i}\mathbb{S}_{i}'\right]^{-1}\frac{1}{\sqrt{N}}\sum_{i=1}^{N}\mathbb{S}_{i}+o_{p}\left(1\right)\label{eq: phi_hat_asym_lin}
\end{equation}
where $\mathbb{S}_{i}=\mathbb{S}_{\phi}\left(\left.X_{i}\right|W_{i};\phi_{0}\right).$
\end{assumption}
Assumption \ref{ass: genalize_p_score} corresponds to a parametric
model for the propensity score when $X$ is a binary treatment indicator.
More generally Assumption \ref{ass: genalize_p_score} requires the
researcher to model the distribution of the policy given controls.
Consider a researcher interested in the relationship between regular
school attendance and student achievement. In this case $Y$ could
be a measure of end-of-school-year achievement, $X$ number of school
days absent, and $W$ a vector of joint determinants of achievement
and attendance (e.g., family background measures, prior achievement,
pre-existing health conditions etc.). In this case the researcher
might assume that the distribution of $X$ given $W$ is Poisson with
\[
\mathbb{E}\left[\left.X\right|W\right]=\exp\left(k\left(W\right)'\phi_{0}\right),\ \mathbb{V}\left(\left.X\right|W\right)=\exp\left(k\left(W\right)'\phi_{0}\right),
\]
where $k\left(W\right)$ is a known $L\times1$ vector of functions
of $W$. Estimation of $\phi_{0}$ , and hence $e\left(W;\phi_{0}\right)$
and $v\left(W;\phi_{0}\right)$, is by maximum likelihood. In most
cases the conditional distribution of $X$ given $W$ can be conveniently
modeled by, depending on the nature of $X$, the appropriate generalized
linear model (GLM). When $X$ is multivariate, the outcome of censoring,
or has mixed discrete/continuous components, then specifying $f\left(\left.x\right|w;\phi\right)$
may involve considerable work. For complicated likelihoods $e\left(W;\hat{\phi}\right)$
and $v\left(W;\hat{\phi}\right)$ may need to be approximated numerically
or by simulation.

\subsection*{The \citet{Wooldridge_CWP04} estimator}

\citet{Wooldridge_CWP04} introduced a two-step estimator for $\beta_{0}$.
A textbook exposition appears in \citet[Chapter 21.6.3 ]{Wooldridge_EACSPDBook10}.
His procedure is summarized in Algorithm \ref{alg: Wooldridge}.

\begin{algorithm}
\caption{\textsc{The \citet{Wooldridge_CWP04} Estimate of $\beta_{0}$\label{alg: Wooldridge} }}

\begin{enumerate}
\item Compute the maximum likelihood estimate of $\phi_{0}$ and construct
$e\left(W_{i},\hat{\phi}\right)$ and $v\left(W_{i},\hat{\phi}\right)$
for $i=1,\ldots,N$;
\item Compute linear instrumental variables fit of $Y$ onto $X$ (with
no constant) using $v\left(W;\hat{\phi}\right)^{-1}\left(X-e\left(W;\hat{\phi}\right)\right)$
as the instrument for $X$. The coefficient on $X$ equals $\hat{\beta}$.
\end{enumerate}
\end{algorithm}

\citet{Wooldridge_CWP04} does not characterize the asymptotic sampling
properties of $\hat{\beta}_{W}$. In this section, we show that Wooldridge's
estimator is not efficient under Assumptions \ref{ass: random_sampling},
\ref{ass: overlap} and \ref{ass: genalize_p_score}. Furthermore
it requires the generalized propensity score to be correctly specified.
The structure of this inefficiency and lack of robustness, as well
as the form of the efficient influence function derived earlier, guides
the construction of our new, locally efficient and doubly robust estimator.

The second step of Algorithm \ref{alg: Wooldridge} corresponds to
finding the $\hat{\beta}_{\mathrm{W}}$ which solves the sample moment
\begin{equation}
\frac{1}{N}\sum_{i=1}^{N}\rho\left(Z_{i},\hat{\phi},\hat{\beta}_{\mathrm{W}}\right)=0,\label{eq: wooldridge_moment}
\end{equation}
for $\rho\left(Z,\phi,\beta\right)=v\left(W;\phi\right)^{-1}\left(X-e\left(W;\phi\right)\right)\left(Y-X'\beta\right).$
Here $\hat{\phi}$ corresponds to the MLE of $\phi_{0}$ computed
in the first step of the procedure. A mean value expansion of (\ref{eq: wooldridge_moment})
in $\hat{\beta}_{W}$ about $\beta_{0}$ yields
\begin{eqnarray*}
\hat{\beta}_{\mathrm{W}} & = & \beta_{0}+\frac{1}{N}\sum_{i=1}^{N}\rho\left(Z,\hat{\phi},\beta_{\mathrm{0}}\right)+o_{p}(N^{-1/2}).
\end{eqnarray*}
Rearrangement of terms and a second mean value expansion in $\hat{\phi}$
about $\phi_{0}$ gives
\begin{eqnarray*}
\sqrt{N}\left(\hat{\beta}_{\mathrm{W}}-\beta_{0}\right) & = & \frac{1}{\sqrt{N}}\sum_{i=1}^{N}\rho\left(Z,\phi_{0},\beta_{\mathrm{0}}\right)\\
 &  & +\left\{ \frac{1}{N}\sum_{i=1}^{N}\frac{\partial\rho\left(Z,\bar{\phi},\beta_{0}\right)}{\partial\phi}\right\} \sqrt{N}\left(\hat{\phi}-\phi_{0}\right)+o_{p}\left(1\right).
\end{eqnarray*}
Observe that under Assumptions \ref{ass: random_sampling} and \ref{ass: overlap}
\begin{eqnarray*}
\mathbb{E}\left[\left.\rho\left(Z,\phi_{0},\beta_{0}\right)\right|W=w\right] & = & \mathbb{E}\left[\left.v\left(W;\phi_{0}\right)^{-1}\left(X-e\left(W;\phi_{0}\right)\right)\left(Y-X'\beta_{0}\right)\right|W=w\right]\\
 & = & b_{0}\left(w\right)-\beta_{0}
\end{eqnarray*}
since $\mathbb{E}\left[\left.v\left(W;\phi_{0}\right)^{-1}\left(X-e\left(W;\phi_{0}\right)\right)X'\right|W=w\right]=I_{K}$.
In integral form: 
\begin{equation}
\int\rho\left(z,\phi_{0},\beta_{0}\right)f_{0}\left(\left.y\right|w,x\right)f\left(\left.x\right|w;\phi_{0}\right)\mathrm{d}x\mathrm{d}y=b_{0}\left(w\right)-\beta_{0}.\label{eq: conditional_mean_of_rho}
\end{equation}
Differentiating (\ref{eq: conditional_mean_of_rho}) through the integral
with respect to $\phi$ gives:
\begin{equation}
\mathbb{E}\left[\left.\frac{\partial\rho\left(Z,\phi_{0},\beta_{0}\right)}{\partial\phi}\right|W=w\right]=-\mathbb{E}\left[\left.\rho\left(Z,\phi_{0},\beta_{0}\right)\mathbb{S}'\right|W=w\right],\label{eq: GIME}
\end{equation}
which is a Generalized Information Matrix Equality (GIME) result \citep[e.g., ][p. 104]{Newey_JAE90}.

Using (\ref{eq: phi_hat_asym_lin}) and (\ref{eq: GIME}) we have
\begin{eqnarray}
\sqrt{N}\left(\hat{\beta}_{\mathrm{W}}-\beta_{0}\right) & = & \frac{1}{\sqrt{N}}\sum_{i=1}^{N}\rho_{i}\nonumber \\
 &  & -\mathbb{E}\left[\rho\mathbb{S}'\right]\mathbb{E}\left[\mathbb{S}\mathbb{S}'\right]^{-1}\frac{1}{\sqrt{N}}\sum_{i=1}^{N}\mathbb{S}_{i}+o_{p}\left(1\right)\nonumber \\
 & = & \frac{1}{\sqrt{N}}\sum_{i=1}^{N}\left\{ \rho_{i}-\mathbb{E}\left[\rho\mathbb{S}'\right]\mathbb{E}\left[\mathbb{S}\mathbb{S}'\right]^{-1}\mathbb{S}_{i}\right\} +o_{p}\left(1\right)\label{eq: Wooldridge_Influence}
\end{eqnarray}
for $\rho_{i}=\rho\left(Z_{i},\phi_{0},\beta_{0}\right).$

Similar to the result of \citet{Wooldridge_JE07} for the binary $X$
case, this asymptotically linear representation of $\hat{\beta}_{\mathrm{W}}$
implies that if practitioners ignore sampling error in $\hat{\phi}$,
they can get conservative confidence intervals. In addition, this
expression shows that over-parameterizing the conditional distribution
of $X$ given $W$ will not decrease the asymptotic precision $\hat{\beta}_{\mathrm{W}}$.

We show next that $\hat{\beta}_{\mathrm{W}}$ is inefficient for $\beta_{0}$
in the semiparametric model defined by Assumptions \ref{ass: random_sampling},
\ref{ass: overlap} and \ref{ass: genalize_p_score}. This demonstration
of inefficiency usefully provides insight into how to construct a
more efficient estimator. We begin by decomposing Wooldridge's \citeyearpar{Wooldridge_CWP04}
identifying moment into the efficient influence function and a remainder:
$\rho\left(Z,\phi_{0},\beta_{0}\right)=\mathrm{\psi}_{\beta}^{\mathrm{eff}}\left(Z,\beta_{0},\phi_{0},h_{0}\left(W\right)\right)+r\left(W,X,\beta_{0},\phi_{0},h_{0}\left(W\right)\right)$
with
\begin{align}
r\left(W,X,\beta_{0},\phi_{0},h_{0}\left(W\right)\right)= & v\left(W;\phi_{0}\right)^{-1}\left(X-e\left(W;\phi_{0}\right)\right)\left(a_{0}\left(W\right)+X'\left(b_{0}\left(W\right)-\beta_{0}\right)\right)\label{eq: r(w,x)}\\
 & -\left(b_{0}\left(W\right)-\beta_{0}\right)\nonumber 
\end{align}
Let $r_{0}\left(W,X\right)=r\left(W,X,\beta_{0},\phi_{0},h_{0}\left(W\right)\right).$
Note that $\mathbb{E}\left[\left.r_{0}\left(W,X\right)\right|W\right]=0.$
Note further that $\mathbb{S}$ is also conditionally mean zero given
$W$.

Now observe that for $l=1,\ldots,\dim\left(\phi\right)$
\begin{eqnarray*}
\frac{\partial\mathrm{\psi}_{\beta}^{\mathrm{eff}}}{\partial\phi_{l}} & = & -v\left(W;\phi_{0}\right)^{-1}\frac{\partial v\left(W;\phi_{0}\right)}{\partial\phi_{l}}v\left(W;\phi_{0}\right)^{-1}\left(X-e\left(W;\phi_{0}\right)\right)U\\
 &  & -v\left(W;\phi_{0}\right)^{-1}\frac{\partial e\left(W;\phi_{0}\right)}{\partial\phi_{l}}U,
\end{eqnarray*}
and hence that
\begin{eqnarray}
\mathbb{E}\left[\left.\frac{\partial\mathrm{\psi}_{\beta}^{\mathrm{eff}}}{\partial\phi_{l}}\right|W\right] & = & -v\left(W;\phi_{0}\right)^{-1}\frac{\partial v\left(W;\phi_{0}\right)}{\partial\phi_{l}}v\left(W;\phi_{0}\right)^{-1}\mathbb{E}\left[\left.\left(X-e\left(W;\phi_{0}\right)\right)U\right|W\right]\label{eq: eff_infl_grad_cond}\\
 &  & -v\left(W;\phi_{0}\right)^{-1}\frac{\partial e\left(W;\phi_{0}\right)}{\partial\phi_{l}}\mathbb{E}\left[\left.U\right|W\right]\nonumber \\
 & = & 0\nonumber 
\end{eqnarray}
by Lemma \ref{lem: Wooldridge} above.

Next start with the fact that
\[
\int\mathrm{\psi}_{\beta}^{\mathrm{eff}}f_{0}\left(\left.y\right|x,w\right)f\left(\left.x\right|w;\phi_{0}\right)f_{0}\left(w\right)=0.
\]
Differentiating through the integral gives the equality
\[
\int\frac{\partial\mathrm{\psi}_{\beta}^{\mathrm{eff}}}{\partial\phi'}f_{0}\left(\left.y\right|x,w\right)f\left(\left.x\right|w;\phi_{0}\right)f_{0}\left(w\right)=-\int\left\{ \mathrm{\psi}_{\beta}^{\mathrm{eff}}\mathbb{S}'\right\} f_{0}\left(\left.y\right|x,w\right)f\left(\left.x\right|w;\phi_{0}\right)f_{0}\left(w\right)
\]
and hence that, using the decomposition of $\rho\left(Z,\phi_{0},\beta_{0}\right)$
introduced above and equation (\ref{eq: eff_infl_grad_cond}),
\[
\mathbb{E}\left[\rho\mathbb{S}'\right]=\mathbb{E}\left[\psi_{\beta}^{\mathrm{eff}}\mathbb{S}'\right]+\mathbb{E}\left[r\mathbb{S}'\right]=\mathbb{E}\left[r\mathbb{S}'\right].
\]
Plugging this into our influence function we get
\begin{eqnarray*}
\sqrt{N}\left(\hat{\beta}_{\mathrm{W}}-\beta_{0}\right) & = & \frac{1}{\sqrt{N}}\sum_{i=1}^{N}\left\{ \rho_{i}-\mathbb{E}\left[\rho\mathbb{S}'\right]\mathbb{E}\left[\mathbb{S}\mathbb{S}'\right]^{-1}\mathbb{S}_{i}\right\} +o_{p}\left(1\right)\\
 & = & \frac{1}{\sqrt{N}}\sum_{i=1}^{N}\left\{ \psi_{\beta,i}^{\mathrm{eff}}+\left[r_{i}-\mathbb{E}\left[r\mathbb{S}'\right]\mathbb{E}\left[\mathbb{S}\mathbb{S}'\right]^{-1}\mathbb{S}_{i}\right]\right\} +o_{p}\left(1\right),
\end{eqnarray*}
and hence an asymptotic distribution of
\begin{equation}
\sqrt{N}\left(\hat{\beta}_{\mathrm{W}}-\beta_{0}\right)\stackrel{D}{\rightarrow}\mathcal{N}\left(0,\mathcal{I}\left(\beta_{0}\right)^{-1}+\mathbb{E}\left[\left(r-\Pi_{r\mathbb{S}}\mathbb{S}\right)\left(r-\Pi_{r\mathbb{S}}\mathbb{S}\right)'\right]\right)\label{eq: wooldridge_avar}
\end{equation}
with $\Pi_{r\mathbb{S}}=\mathbb{E}\left[r\mathbb{S}'\right]\times\mathbb{E}\left[\mathbb{S}\mathbb{S}'\right]^{-1}.$ 

The form of the the limit distribution (\ref{eq: wooldridge_avar})
is similar to that of the familiar inverse probability weighting (IPW)
estimator for binary treatments \citep[e.g.,][Proposition 3.1]{Graham_Pinto_Egel_ReStud12}.
In that context it is well-known that replacing a known propensity
score with an estimated one increases precision \citep{Hirano_et_al_EM03,Hitomi_et_al_ET08,Graham_EM11}.
In principle the degree of precision increase is increasing in the
complexity/richness of the fitted propensity score model. Expression
(\ref{eq: wooldridge_avar}) indicates that a similar phenomena operates
in our setting. If the portion of the efficient influence function
that is omitted by the \citet{Wooldridge_CWP04} procedure is well-approximated
by a linear combination of the scores used to estimate the propensity
score, then the $\hat{\beta}_{\mathrm{W}}$ will be precisely determined.
In practice, instead of relying on a possibly overfitted propensity
score to yield efficient estimates, it is better to redesign the estimation
procedure with efficiency in mind at the outset. 

\subsection*{A locally efficient, doubly robust estimator}

Our estimator for $\beta_{0}$ requires a working parametric model
for the CLP coefficients $a_{0}\left(W\right)$ and $b_{0}\left(W\right)$.
Consistency and asymptotic normality of our estimate, $\hat{\beta}$,
will not depend on the correctness of this working model, but its
limiting variance will. A convenient working model is provided by
Assumption \ref{ass: clp_coefficients}.
\begin{assumption}
\label{ass: clp_coefficients}\textsc{(CLP Coefficients):} $a_{0}\left(W\right)=\alpha_{0}+\left(W-\mu_{W}\right)'\gamma_{0}$
and $b_{0}\left(W\right)=\beta_{0}+\Delta_{0}\left(W-\mu_{W}\right)$.
\end{assumption}
In practice these models for $a_{0}\left(W\right)$ and $b_{0}\left(W\right)$
can be made arbitrarily flexible since $W$ can include a rich set
of basis functions (e.g., squares, cross-products etc.) in the underlying
controls.

Under Assumption \ref{ass: clp_coefficients} we have that
\begin{eqnarray}
\mathbb{E}^{*}\left[\left.Y\right|X;W\right] & = & \alpha_{0}+\left(W-\mu_{W}\right)'\gamma_{0}+X'\left(\beta_{0}+\Delta_{0}\left(W-\mu_{W}\right)\right)\nonumber \\
 & = & \alpha_{0}+\left(W-\mu_{W}\right)'\gamma_{0}+\left(\left(W-\mu_{W}\right)\varotimes X\right)'\delta_{0}+X'\beta_{0},\label{eq: oaxaca-blinder}
\end{eqnarray}
where $\delta_{0}=\mathrm{vec}\left(\Delta_{0}\right).$ 

Equation (\ref{eq: oaxaca-blinder}) implies that, maintaining Assumption
\ref{ass: clp_coefficients}, one approach to estimating $\beta_{0}$
would be to compute the least squares fit of $Y_{i}$ onto a constant,
$W_{i}-\mu_{W}$, all interactions of $W_{i}-\mu_{i}$ and $X_{i}$,
and $X_{i}$ itself. For the special case where $X$ is a binary treatment
indicator, this estimator is familiar to labor economists as a Oaxaca-Blinder
average treatment effect (ATE) estimator \citep[e.g.,][]{Sloczynski_OBES15}.\footnote{In this literature researchers typically center $W$ around $\mathbb{E}\left[\left.W\right|X=1\right]$,
not the unconditional mean $\mu_{W}=\mathbb{E}\left[W\right]$ as
is done here. With this alternative centering the coefficient on $X_{i}$
will coincide with the average treatment effect on the treated (ATT).} Consistency of of this estimator hinges upon Assumption \ref{ass: clp_coefficients}
accurately characterizing the sampled population. 

In our setting Assumption \ref{ass: clp_coefficients} plays a different
role. Unlike in the Oaxaca-Blinder procedure, its validity is not
required for consistency, but if it does accurately described the
sampled population our estimator will be highly efficient. These benefits
come at the cost of assuming that prior knowledge regarding the form
of the generalized propensity score is available (i.e., maintaining
Assumption \ref{ass: genalize_p_score}).

To describe our procedure we require some additional notation. Let
$\lambda=\left(\alpha,\gamma',\delta'\right)'$, $R\left(\mu_{W}\right)=\left(1,\left(W-\mu_{W}\right)',\left(\left(W_{i}-\mu_{W}\right)\varotimes X_{i}\right)'\right)'$
and
\[
U_{i}\left(\mu_{W},\lambda,\beta\right)=\left(Y_{i}-R\left(\mu_{W}\right)'\lambda-X_{i}'\beta\right).
\]
When $R\left(\mu_{W}\right)$ is evaluated at the correct population
mean of $W$, we often simply write $R$.

Our estimator is based upon the $\left(L+J+1+J+JK+K\right)\times1$
vector of moment conditions, $m\left(Z_{i},\theta\right)$, partitioned
into the three ordered sub-vectors:
\begin{align}
\underset{L\times1}{m_{1}(X_{i},W_{i},\phi)}= & \mathbb{S}_{\phi}\left(\left.X_{i}\right|W_{i};\phi\right)\label{eq: m1_p1}\\
\underset{J\times1}{m_{2}(W_{i},\mu_{W})}= & W_{i}-\mu_{W}\label{eq: m2_p1}\\
\underset{1+J+JK+K\times1}{m_{3}(Z_{i},\phi,\mu_{W},\lambda,\beta)}= & \left(\begin{array}{c}
R_{i}\left(\mu_{W}\right)\\
v\left(W;\phi\right)^{-1}\left(X-e\left(W;\phi\right)\right)
\end{array}\right)U_{i}\left(\mu_{W},\lambda,\beta\right)\label{eq: m3_p1}
\end{align}
where $\theta=\left(\phi,\mu_{W},\lambda',\beta\right)'$ with $\dim\left(\theta\right)=L+J+1+J+JK+K$. 

Equations (\ref{eq: m1_p1}), (\ref{eq: m2_p1}) and (\ref{eq: m3_p1})
constitute a just-identified system. The corresponding method-of-moments
estimate of $\beta_{0}$ can be computed in the three simple steps
listed in Algorithm \ref{alg: Estimation}.

\begin{algorithm}

\caption{\label{alg: Estimation}\textsc{Locally Efficient and Doubly Robust
Estimation Of $\beta_{0}$}}

\begin{enumerate}
\item Compute the maximum likelihood estimate of $\phi_{0}$ and construct
$e\left(W_{i},\hat{\phi}\right)$ and $v\left(W_{i},\hat{\phi}\right)$
for $i=1,\ldots,N$;
\item Compute the sample mean $\hat{\mu}_{W}=\frac{1}{N}\sum_{i=1}^{N}W_{i}$
and construct $R_{i}\left(\hat{\mu}_{W}\right)$ for $i=1,\ldots,N$;
\item Compute the linear instrumental variables fit of $Y_{i}$ onto $R_{i}\left(\hat{\mu}_{W}\right)$
and $X_{i}$ using $v\left(W_{i};\hat{\phi}\right)^{-1}\left(X_{i}-e\left(W_{i};\hat{\phi}\right)\right)$
as the excluded instrument for $X_{i}$. The coefficient on $X_{i}$
in this fit coincides with $\hat{\beta}$.
\end{enumerate}
\end{algorithm}

In many cases of interest Algorithm \ref{alg: Estimation} is easily
implemented using standard software. Standard errors may be constructed
in the usual way for GMM estimators \citep[e.g.,][]{Newey_McFadden_HBE94,Wooldridge_EACSPDBook10}
or using a bootstrap. 

In step 3, if instead we let $X_{i}$ serve as its own instrument,
we get an ``Oaxaca-Blinder'' type estimator.

The next theorem summarizes the large sample properties of $\hat{\beta}$.
In the statement of the Theorem, $\Delta_{*}$ denotes the limiting
pseudo-true value of $\hat{\Delta}$. If Assumption \ref{ass: clp_coefficients}
additionally holds then $\Delta_{*}=\Delta_{0}$. We also define $\tilde{\epsilon}=v\left(W\right)_{0}^{-1}\left(X-e_{0}\left(W\right)\right)\epsilon$
where $\epsilon=\left\{ a_{0}\left(W\right)+X'\left(b_{0}\left(W\right)-\beta_{0}\right)-R'\lambda_{*}\right\} $
(with $\lambda_{*}$ denoting a pseudo-true parameter value). Finally
we let $\Pi_{\tilde{\epsilon}\mathbb{S}}=\mathbb{E}\left[\tilde{\epsilon}_{i}\mathbb{S}'\right]\mathbb{E}\left[\mathbb{S}\mathbb{S}'\right]^{-1}$
denote the coefficient matrix associated with the multivariate regression
of $\tilde{\epsilon}$ onto the score vector associated with $\phi_{0}$
(the parameter indexing the generalized propensity score).
\begin{thm}
\textsc{(Large Sample Distribution)}\label{thm: Large-Sample} Consider
the semiparametric problem established by Definition \ref{def: clp}
and Assumptions \ref{ass: random_sampling}, \ref{ass: overlap},
and \ref{ass: genalize_p_score}. Let $\hat{\beta}$ be the method
of moments estimate of $\beta_{0}$ based upon restrictions (\ref{eq: m1_p1})
to (\ref{eq: m3_p1}). Under regularity conditions \citep[cf.,][Theorem 3.4]{Newey_McFadden_HBE94}
$\hat{\beta}$ is (i) asymptotically normal with a limiting distribution
of 
\begin{equation}
\sqrt{N}\left(\hat{\beta}-\beta_{0}\right)\overset{D}{\rightarrow}\mathcal{N}\left(0,\mathbb{E}\left[\Omega_{0}\left(W\right)\right]+\Delta_{*}\mathbb{V}\left(W\right)\Delta_{*}'+\mathbb{E}\left[\left(\tilde{\epsilon}-\Pi_{\tilde{\epsilon}\mathbb{S}}\mathbb{S}\right)\left(\tilde{\epsilon}-\Pi_{\tilde{\epsilon}\mathbb{S}}\mathbb{S}\right)'\right]\right),\label{eq: large_sample_basic}
\end{equation}
and (ii) locally efficient for $\beta_{0}$ at Assumption \ref{ass: clp_coefficients}
with 
\begin{equation}
\sqrt{N}\left(\hat{\beta}-\beta_{0}\right)\overset{D}{\rightarrow}\mathcal{N}\left(0,\mathcal{I}(\beta_{0})^{-1}\right).\label{eq: large_sample_loceff}
\end{equation}
\end{thm}
\begin{proof}
See Appendix \ref{appendix: proofs}.
\end{proof}
Part (ii) of Theorem (\ref{thm: Large-Sample}) follows easily from
part (i). In the proof we show that $\epsilon$ equals the prediction
error associated with the mean squared error minimizing linear prediction
of $a_{0}\left(W\right)+X'\left(b_{0}\left(W\right)-\beta_{0}\right)$
given $R\left(\mu_{W}\right)$. When Assumption \ref{ass: clp_coefficients}
additional holds this prediction error will be identically equal to
zero and the third term in the variance expressing appearing in part
(i) drops out. Similarly when Assumption \ref{ass: clp_coefficients}
holds we have $\Delta_{*}\mathbb{V}\left(W\right)\Delta_{*}'=\mathbb{V}\left(b_{0}\left(W\right)\right)$.
Together these two observations give part (ii).

Our efficiency bound calculation, Theorem \ref{thm: SEB}, gives the
information bound for $\beta_{0}$ without imposing the additional
auxiliary Assumption \ref{ass: clp_coefficients}. This assumption
imposes restrictions on the joint distribution of the data not implied
by the baseline model. If these restrictions are added to the prior
used to calculate the efficiency bound, then it will generally be
possible to estimate $\beta_{0}$ more precisely. Our estimator is
not efficient with respect to this augmented model. Rather it attains
the bound provided by Theorem \ref{thm: SEB} if Assumption \ref{ass: clp_coefficients}
``happens to be true'' in the sampled population, but is not part
of the prior restriction used to calculate the bound. \citet[p. 114]{Newey_JAE90}
discusses the concept of local efficiency in detail. In what follows
we will, for brevity, say $\hat{\beta}$ is locally efficient at Assumption
\ref{ass: clp_coefficients}.

Even if Assumption \ref{ass: clp_coefficients} does not hold precisely,
our procedure will be ``nearly'' efficient when it is approximately
true (in which case variability in $\epsilon$ about zero is small).
A caveat to this claim is that the third variance term in (\ref{eq: large_sample_basic})
may still be large in this case if $v_{0}\left(w\right)$ is nearly
zero for enough values of $w$. This occurs when overlap is poor,
or there exists a lack of variation in the policy variable for some
subpopulations defined in terms of $W=w$. \citet{Graham_Pinto_Egel_JBES16}
develop this observation more extensively for the special case where
$X$ is binary, but similar issues apply in the more general setting
considered here.

Our next result formalizes the above observation. It extends our local
efficiency result to ``near'' global efficiency. The basic argument
mirrors that given by \citet[Proposition 2]{Chamberlain_JE87} for
approximately efficient estimation of conditional moment problems.
Presenting this result requires defining a sequence of estimators
based upon Algorithm \ref{alg: Estimation}.

Let $\mathcal{L}_{2}$ be the space of functions $f\thinspace:\thinspace\mathbb{W}\rightarrow\mathbb{R}$
with finite second moment $\mathbb{E}\left[f\left(W\right)^{2}\right]<\infty$.
Under Assumptions \ref{ass: random_sampling} and \ref{ass: overlap}
the set of feasible conditional linear predictor coefficients lies
within this space such that $a\thinspace:\thinspace W\rightarrow\mathbb{R}^{1}$
and $b\thinspace:\thinspace W\rightarrow\mathbb{R}^{K}$ with $\mathbb{E}\left[a\left(W\right)^{2}\right]<\infty$
and $\mathbb{E}\left[\left\Vert b\left(W\right)\right\Vert ^{2}\right]<\infty.$
Let $\left\{ k_{j}\left(W\right)\right\} _{j=1}^{\infty}$ be a sequence
of linearly independent functions of the control variables, each with
finite variance. Similar to \citet{Chamberlain_JE87} we call this
sequence complete if, (i) for any $\zeta>0$ and (ii) any feasible
conditional linear predictor coefficients $a\left(W\right)$ and $b\left(W\right)$
in $\mathcal{L}_{2}$, there are the real numbers $\alpha,\gamma_{1},\ldots,\gamma_{J}$
and $\delta_{k1},\ldots,\delta_{kJ}$ for $k=1,\ldots,K$ such that
\begin{equation}
\mathbb{E}\left[\left\Vert \delta^{\left(J\right)}\left(W\right)\right\Vert ^{2}\right]<\zeta^{2},\label{eq: L2-complete}
\end{equation}
with $\delta^{\left(J\right)}\left(W\right)$ defined as
\begin{equation}
\delta^{\left(J\right)}\left(W\right)=\left(\begin{array}{c}
a\left(W\right)-\alpha-\sum_{j=1}^{J}\left(k_{j}\left(W\right)-\mu_{j}\right)\gamma_{j}\\
b_{1}\left(W\right)-\beta_{01}-\sum_{j=1}^{J}\left(k_{j}\left(W\right)-\mu_{j}\right)\delta_{1j}\\
\vdots\\
b_{K}\left(W\right)-\beta_{0K}-\sum_{j=1}^{J}\left(k_{j}\left(W\right)-\mu_{j}\right)\delta_{Kj}
\end{array}\right).\label{eq: delta(W)_J}
\end{equation}

Let $k^{\left(J\right)}\left(W\right)$ be the $J\times1$ vector
of functions of $W$ with $j^{th}$ element $k_{j}\left(W\right)$.
We can construct a sequence of estimators, indexed by $J$, based
upon Algorithm \ref{alg: Estimation} with $k^{\left(J\right)}\left(W\right)$
replacing $W$. To do this let $\mu^{\left(J\right)}=\mathbb{E}\left[k^{\left(J\right)}\left(W\right)\right]$
and additionally define
\[
R^{\left(J\right)}=\left(1,\left(k^{\left(J\right)}\left(W\right)-\mu^{\left(J\right)}\right)',\left(\left(k^{\left(J\right)}\left(W\right)-\mu^{\left(J\right)}\right)\otimes X\right)'\right)'.
\]
We can then estimate $\beta_{0}$ by Algorithm \ref{alg: Estimation}
with $k^{\left(J\right)}\left(W\right)$, $\mu^{\left(J\right)}$
and $R^{\left(J\right)}$ respectively replacing $W$, $\mu_{W}$,
and $R\left(\mu_{W}\right)$. 

Consider the asymptotic precision matrix of this method of moments
estimator; from Theorem \ref{thm: Large-Sample} we get 
\begin{align*}
\mathcal{I}^{\left(J\right)}\left(\beta_{0}\right)^{-1}= & \mathbb{E}\left[\Omega_{0}\left(W\right)\right]+\Delta_{*}^{\left(J\right)}\mathbb{V}\left(k^{\left(J\right)}\left(W\right)\right)\left(\Delta_{*}^{\left(J\right)}\right)'\\
 & +\mathbb{E}\left[\left(\tilde{\epsilon}^{\left(J\right)}-\Pi_{\tilde{\epsilon}\mathbb{S}}^{\left(J\right)}\mathbb{S}\right)\left(\tilde{\epsilon}^{\left(J\right)}-\Pi_{\tilde{\epsilon}\mathbb{S}}^{\left(J\right)}\mathbb{S}\right)'\right].
\end{align*}
with $\mathcal{I}^{\left(J\right)}\left(\beta_{0}\right)^{-1}\geq\mathcal{I}\left(\beta_{0}\right)^{-1}$
(here ``$A\geq B$'' denotes ``$A-B$ is positive semi-definite'').
Recall that $\mathcal{I}\left(\beta_{0}\right)$ is the semiparametric
efficiency bound given in Theorem \ref{thm: SEB}. Let $\hat{\beta}^{\left(J\right)}$
denote the estimate of $\beta_{0}$ based upon $k^{\left(J\right)}\left(W\right)$.
\begin{prop}
\textsc{(Near Efficiency) }\label{prop: near_efficiency}If, maintaining
the Assumptions of part (i) of Theorem \ref{thm: Large-Sample}, $\left\{ \hat{\beta}^{\left(J\right)}\right\} $
is based upon a linearly independent, complete sequence $\left\{ k_{j}\left(W\right)\right\} _{j=1}^{\infty},$
then, for $\mathbb{X}\times\mathbb{W}$ a compact subset of $\mathbb{R}^{K+\dim\left(W\right)}$,
\[
\underset{J\rightarrow\infty}{\lim}\mathcal{I}^{\left(J\right)}\left(\beta_{0}\right)^{-1}=\mathcal{I}\left(\beta_{0}\right)^{-1}.
\]
\end{prop}
\begin{proof}
See Appendix \ref{appendix: proofs}.
\end{proof}
The compact support assumption invoked in the statement of the theorem
is used in the proof, but does not appear to be essential.

Proposition \ref{prop: near_efficiency} leaves unanswered important
practical questions, such as how quickly $J$ should increase with
$N$. More generally the question of exactly how to choose the elements
of $k^{\left(J\right)}\left(W\right)$ in order to achieve good precision
in practice remains unanswered. However we expect that many insights
from related settings could be applied here \citep[e.g.,][]{Belloni_et_al_ReStud14}.

We conclude this section by demonstrating double robustness in the
sense of \citet{Scharfstein_Rotnitzky_Robins_JASA99b}. If the specification
of the generalized propensity score is not correct (i.e., Assumption
\ref{ass: genalize_p_score} does not hold), but Assumption \ref{ass: clp_coefficients}
is true, then our estimator remains consistent for $\beta_{0}$. Recall
that Assumption \ref{ass: clp_coefficients} was initially invoked
to ensure local efficiency of our procedure. It turns out that modeling
the form of the conditional linear predictor coefficients has the
added benefit of ensuring that our estimator remains consistent even
if our generalized propensity score model is incorrect. Double robustness
results are familiar from the literature on missing data and program
evaluation \citep[e.g.,][]{Scharfstein_Rotnitzky_Robins_JASA99b,Cattaneo_JOE10,Graham_EM11}.
In these settings $X$ is binary or a vector of mutually-exclusive
treatment indicators. Double robustness in our more general setting
is perhaps unsurprising, but nevertheless a new result.

To understand this result observe that step 3 of Algorithm \ref{alg: Estimation}
corresponds to solving the sample analog of
\[
\mathbb{E}\left[\left(\begin{array}{c}
R\left(\mu_{W}\right)\\
v\left(W;\phi_{*}\right)^{-1}\left(X-e\left(W;\phi_{*}\right)\right)
\end{array}\right)U\left(\mu_{W},\lambda_{0},\beta_{0}\right)\right]=0
\]
for $\lambda_{0}$ and $\beta_{0}$. Here we use the notation $\phi_{*}$
to denote that our generalized propensity score model may be miss-specified. 

If Assumption \ref{ass: clp_coefficients} holds in the population,
then $U_{0}=U_{i}\left(\mu_{W},\lambda_{0},\beta_{0}\right)$ is a
conditional linear predictor (CLP) error and Lemma \ref{lem: Wooldridge}
above applies. Recall that $R\left(\mu_{W}\right)=\left(1,\left(W-\mu_{W}\right)',\left(\left(W_{i}-\mu_{W}\right)\varotimes X_{i}\right)'\right)'$;
by Lemma \ref{lem: Wooldridge} $U_{0}$ is uncorrelated with all
components of this vector. Likewise, because $U_{0}$ is mean independent
of $W$ and conditionally uncorrelated with $X$, we also have that
$\mathbb{E}\left[v\left(W;\phi_{*}\right)^{-1}\left(X-e\left(W;\phi_{*}\right)\right)U\right]$
is mean zero as well. Hence step 3 of Algorithm \ref{alg: Estimation}
involves the computation of a correctly specified method-of-moments
estimator under Assumption \ref{ass: clp_coefficients}; irrespective
of whether Assumption \ref{ass: genalize_p_score} additionally holds.
Double robustness follows, more or less, directly.

The above discussion also clarifies why, as is sometimes true in practice,
sampling variability in our estimator can theoretically be lower than
the semiparametric variance bound in Theorem \ref{thm: SEB} when
the generalized propensity score is misspecified, but the form of
the CLP coefficients are not. First, recall that the variance bound
is computed without making any a priori assumptions about the form
of the CLP coefficients. It turns out that in our setting such assumptions
generally increase the precision with which $\beta_{0}$ may be estimated.
When we invoke the double robustness property of our procedure to
ensure consistency we are in a situation where the veracity of Assumption
\ref{ass: clp_coefficients} is central. Whereas is the setting covered
by Theorem \ref{thm: Large-Sample}, Assumption \ref{ass: clp_coefficients}
``may happen to be true'', but need not be.

It is instructive to compare our estimator with the ``Oaxaca-Blinder-type''
one described earlier. The Oaxaca-Blinder procedure necessarily maintains
Assumption \ref{ass: clp_coefficients}. Since this restriction is
part of the prior, it would not be surprising to find that, under
correct specification, that the Oaxaca-Blinder estimator is more efficient
than ours. For the purposes of developing this point, additionally
assume that $U_{0}$ is homoscedastic in $X$ and $W$ (but that this
is not part of the prior), then \textendash{} maintaining Assumption
\ref{ass: clp_coefficients} \textendash{} replacing $v\left(W;\phi_{*}\right)^{-1}\left(X-e\left(W;\phi_{*}\right)\right)$
with $X$ in the above moment would be natural. This replacement leads
the researcher to the Oaxaca-Blinder estimator (which will also be
efficient in this case). Hence, when Assumption \ref{ass: clp_coefficients}
does hold in the sampled population, our procedure will be less efficient
that the Oaxaca-Blinder one (at least under homoscedasticity of $U_{0}$).
Of course, when Assumption \ref{ass: clp_coefficients} does not characterize
the sampled population, our procedure remains consistent, while the
Oaxaca-Blinder one does not.
\begin{thm}
(\textsc{Double Robustness})\label{thm:DoublyRobust} Under Assumptions
\ref{ass: random_sampling} and \ref{ass: overlap} , $\hat{\beta}\overset{p}{\rightarrow}\beta_{0}$
if either Assumption \ref{ass: genalize_p_score} or \ref{ass: clp_coefficients}
holds.
\end{thm}
The proof is straightforward and omitted (see \citet{Graham_Pinto_Egel_ReStud12}
and \citet{Graham_Pinto_Egel_JBES16} for proofs of related results).
As a practical matter using the standard method-of-moments sandwich
variance-covariance matrix estimator associated with the moment problem
defined by (\ref{eq: m1_p1}), (\ref{eq: m2_p1}) and (\ref{eq: m3_p1})
above will support asymptotically valid inference under the conditions
of both Theorems \ref{thm: Large-Sample} and \ref{thm:DoublyRobust}.

\section{Examples and special cases\label{sec: examples}}

In this section we demonstrate that our semiparametric regression
model encompasses several other well-known models.

\subsection*{Example 1: Binary Treatment Effect}

Following the program evaluation literature let $Y_{0}$ denote the
potential outcome under control and $Y_{1}$ the potential outcome
under active treatment treatment. For each sampled unit we observe
either $Y_{0}$ or $Y_{1}$ but not both. The observed outcome, $Y$,
therefore equals
\[
Y=XY_{1}+(1-X)Y_{0}
\]
where $X$ equals 1 if the unit is treated and zero otherwise. Rewriting
yields a random coefficients model of
\[
Y=A+BX
\]
with $A=Y_{0}$ and $B=Y_{1}-Y_{0}$. The average treatment effect
(ATE) equals
\[
\beta_{0}=\mathbb{E}[Y_{1}-Y_{0}]=\mathbb{E}[B].
\]

\citet{Rosenbaum_Rubin_BM83} show that the ATE is identified when
$(Y_{0},Y_{1})\bot X|W$ (unconfoundedness) and $0<\Pr\left(\left.X=1\right|W=w\right)<1$
for all $w\in\mathbb{W}$ (overlap).

When $X$ is binary our Assumption \ref{ass: CondExog} implies unconfoundedness.
Assumption \ref{ass: CondExog} implies that $X$ is conditionally
uncorrelated with the two potential outcomes. When $X$ is binary
this also corresponds to mean and full conditional independence. Next
observe that $e_{0}(W)=\Pr\left(\left.X=1\right|W=w\right)$ and $v_{0}(W)=e_{0}(W)\left[1-e_{0}(W)\right]$.
Hence our Assumption \ref{ass: overlap} implies that $0<\kappa\leq e_{0}\left(W\right)\leq1-\kappa<1$
or so called strong overlap.

Now consider Algorithm \ref{alg: Estimation}. When $X$ is binary
we have that
\[
v\left(W,\hat{\phi}\right)^{-1}\left(X-e\left(W,\hat{\phi}\right)\right)=\frac{X}{e\left(W,\hat{\phi}\right)}-\frac{1-X}{1-e\left(W,\hat{\phi}\right)}.
\]
Our ATE estimate is the coefficient on $X$ associated with the linear
instrumental variables fit of $Y$ onto a constant, $(W-\hat{\mu}_{W})$,
$(W-\hat{\mu}_{W})\cdot X$, and $X$ where $\frac{X}{e\left(W,\hat{\phi}\right)}-\frac{1-X}{1-e\left(W,\hat{\phi}\right)}$
serves as an instrument for $X$. This estimator is similar to, but
distinct from, the weighted least squares (WLS) one introduced by
\citet{Hirano_Imbens_HSORM01}.

\citet{Wooldridge_CWP04} shows, for $X$ binary, that equation (\ref{eq: generalized_ipw})
coincides with 
\begin{align*}
\mathbb{E}\left[v_{0}\left(W\right)^{-1}\left(X-e_{0}\left(W\right)\right)Y\right] & =\mathbb{E}\left[\frac{XY}{e_{0}\left(W\right)}-\frac{\left(1-X\right)Y}{1-e_{0}\left(W\right)}\right],
\end{align*}
which is the familiar inverse probability weighting (IPW) representation
of the average treatment effect (ATE) in, for example, \citet{Hirano_et_al_EM03}. 

The general form of the efficient influence function given in Theorem
\ref{thm: SEB} above corresponds to the specialized one for the ATE
when $X$ is binary derived by, for example, \citet{Hahn_EM98} and
\citet{Hirano_et_al_EM03}. Hence our general procedure, as summarized
by Algorithm \ref{alg: Estimation}, provides a locally efficient
and double robust estimator of the ATE. To the best of our knowledge,
our proposed estimator is a new one, even in the special case where
it identifies the ATE of a binary treatment. \citet{Bang_Robins_BM05}
and \citet{Tsiatis_Book06} provide introductions to double robust
causal inference. 

\subsection*{Example 2: Multiple Treatment Effects}

Following \citet{Imbens_BM00}, \citet{Wooldridge_CWP04}, and \citet{Cattaneo_JOE10}
consider finite collection of mutually exclusive treatments indexed
by $k\in\{0,1,2,...,K\}$ with $K\in\mathbb{N}$. Associated with
these treatments are the $K+1$ potential outcomes, $Y(0),Y(1),\ldots,Y(K)$.
The observed outcome is
\[
Y=Y(0)+\sum_{k=1}^{K}X_{k}\left\{ Y(k)-Y(0)\right\} 
\]
where $X_{k}$ is a binary random variable that equals $1$ if treatment
$k=0,\ldots,K$ is assigned to the unit and zero otherwise. In this
case, we work with the following random coefficient model:
\[
Y=A+X^{'}B
\]
where $X=\left(X_{1},\ldots,X_{K}\right)'$ is a $K\times1$ vector
of treatment indicators and $B$ a corresponding vector of individual
treatment effects.

In this setup $X$ is multinomial with a conditional mean of
\[
e_{0}\left(W\right)=\left(\begin{array}{c}
\Pr\left(\left.X_{1}=1\right|W\right)\\
\vdots\\
\Pr\left(\left.X_{K}=1\right|W\right)
\end{array}\right)
\]
and an inverse conditional variance of \citep[cf.,][]{Henderson_Searle_SIAM81}
\begin{align*}
v_{0}\left(W\right)^{-1}= & diag\left\{ \frac{1}{\Pr\left(\left.X_{1}=1\right|W\right)},\ldots,\frac{1}{\Pr\left(\left.X_{K}=1\right|W\right)}\right\} \\
 & +\frac{1}{1-\sum_{k=1}^{K}\Pr\left(\left.X_{K}=1\right|W\right)}\iota_{K}\iota_{K}'.
\end{align*}
A little bit of tedious algebra then gives
\[
\beta_{0}=\mathbb{E}\left[v_{0}\left(W\right)^{-1}\left(X-e_{0}\left(W\right)\right)Y\right]=\mathbb{E}\left[\begin{array}{c}
\frac{X_{1}Y}{\Pr\left(\left.X_{1}=1\right|W\right)}-\frac{X_{0}Y}{\Pr\left(\left.X_{0}=1\right|W\right)}\\
\vdots\\
\frac{X_{K}Y}{\Pr\left(\left.X_{K}=1\right|W\right)}-\frac{X_{0}Y}{\Pr\left(\left.X_{0}=1\right|W\right)}
\end{array}\right],
\]
which corresponds to the IPW representation of the ATEs 
\[
\beta_{0}=\left(\begin{array}{c}
\mathbb{E}\left[Y\left(1\right)-Y\left(0\right)\right]\\
\vdots\\
\mathbb{E}\left[Y\left(K\right)-Y\left(0\right)\right]
\end{array}\right),
\]
in the multiple treatment setting.

As in the case where $X$ is binary, the general form of the efficient
influence function given in Theorem \ref{thm: SEB} above corresponds
to the specialized one derived by \citet{Cattaneo_JOE10}. Hence our
general procedure also provides a locally efficient and doubly robust
estimate of ATEs in the multiple, mutually exclusive, treatments setting.

\subsection*{Example 3: Partially linear model}

\citet{Chamberlain_WP86,Chamberlain_EM92} and \citet{Robinson_EM88}
studied the semiparametric partially linear regression model (PLM)
\[
Y=X'\beta_{0}+h_{0}(W)+U
\]
with $\mathbb{E}[U|W,X]=0.$ This model can be represented by the
random coefficient model
\[
Y=A+X'B
\]
where $\mathbb{E}[A|W,X]=a_{0}\left(W\right)=h_{0}(W)$ and $\mathbb{E}[B|X,W]=b_{0}\left(W\right)=\beta_{0}$.
These assumptions are stronger than those implied by Assumption \ref{ass: CondExog}. 

To fit this into our framework we replace Assumption \ref{ass: clp_coefficients}
with a working CLP model of
\[
a_{0}\left(W\right)=h_{0}\left(W\right)=\alpha_{0}+\left(W-\mu_{W}\right)'\lambda_{0},\thinspace\thinspace\thinspace b_{0}\left(W\right)=\beta_{0}.
\]
This implies a constant additive treatment effect structure.

Estimation follows Algorithm \ref{alg: Estimation}. First compute
the MLE of $\phi_{0}$. Second compute the sample means $\hat{\mu}_{W}=\frac{1}{N}\sum_{i=1}^{N}W_{i}$.
Finally compute the linear instrumental variables fit of $Y$ onto
a constant, $\left(W-\hat{\mu}_{W}\right)$ and $X$, using $v\left(W,\hat{\phi}\right)^{-1}\left(X-e\left(W,\hat{\phi}\right)\right)$
as an instrument for $X$. Because of the constant additive treatment
effect structure of the PLM we exclude the interactions $\left(W-\hat{\mu}_{W}\right)\otimes X$
from the IV fit computed in the third step.

It is important to recognize that although our procedure invokes the
working assumption that the treatment effect is constant in $W$ (i.e.,
$b_{0}\left(w\right)=\beta_{0}$ for all $w\in\mathbb{W}$), this
assumption is not required for consistency as long as our generalized
propensity score model is correct. Put differently although our procedure
incorporates the PLM structure, this structure is not part of the
maintained prior (albeit the form of the generalized propensity score
is part of the prior).

If $b_{0}\left(w\right)=\beta_{0}$ for all $w\in\mathbb{W}$ and
$U$ is conditionally mean zero given \emph{both} $W$ and $X$ (and
also has a constant variance), but these are not part of the prior
restriction used to calculate the bound, then (\ref{eq: Inverse_Information_CRC})
evaluates to 
\[
\mathcal{I}(\beta_{0})^{-1}=\sigma^{2}\mathbb{E}\left[v_{0}(W)^{-1}\right].
\]
The modified PLM estimator described above, and based on our Algorithm
\ref{alg: Estimation}, will attain this bound when the true model
is a partially linear one.

\citet{Chamberlain_EM92} gives a bound for $\beta_{0}$ \textendash{}
where the partially linear regression structure \emph{is} part of
the prior restriction (but the homoscedasticity assumption is not)
\textendash{} of
\[
\mathcal{I}_{\mathrm{plm}}(\beta_{0})^{-1}=\sigma^{2}\mathbb{E}\left[v_{0}(W)\right]^{-1}.
\]
The difference $\mathcal{I}(\beta_{0})^{-1}-\mathcal{I}_{\mathrm{plm}}(\beta_{0})^{-1}$
is positive semi-definite. This follows directly from, for example,
the Theorem in \citet{Groves_Rothenberg_BM69} on the expectations
of inverse matrices. Hence although our approach to estimation remains
consistent for $\beta_{0}$ when the true regression function takes
a partially linear form, it will generally be less efficient than
methods which exploit this structure at the outset \citep[e.g.,][]{Robinson_EM88,Robins_Mark_Newey_BM92}.

\section{\label{sec: Monte_Carlo}Finite sample properties}

In order to assess the approximation accuracy of Theorems \ref{thm: Large-Sample}
and \ref{thm:DoublyRobust} in finite samples we conducted a small
simulation experiment, the results of which we report here. We considered
four designs. The outcome was generated according to
\[
Y=a_{0}\left(W\right)+b_{0}\left(W\right)X+U
\]
with $W$ and $U$ independent standard normal random variables and
$a_{0}\left(W\right)$ and $b_{0}\left(W\right)$ either linear (designs
1 and 2) or quadratic (designs 3 and 4) in $W$. The conditional distribution
of $X$ given $W$ was specified as Poisson with parameter $\exp\left(k\left(W\right)'\phi\right)$
and $k\left(W\right)=\left(1,W\right)'$ in designs 1 and 3 and $k\left(W\right)=\left(1,W,W^{2}\right)'$
in designs 2 and 4. Complete details on the data generating process
are given in Table \ref{tab: Monte-Carlo-Designs}.

\begin{table}

\caption{\label{tab: Monte-Carlo-Designs}Monte Carlo Designs}

\begin{centering}
\begin{tabular}{|c|c|c|c|c|}
\hline 
Designs &
1 &
2 &
3 &
4\tabularnewline
\hline 
\hline 
$\alpha_{0}$ &
1 &
1 &
1.5 &
1.5\tabularnewline
\hline 
$\gamma_{1}$ &
1 &
1 &
1 &
1\tabularnewline
\hline 
$\gamma_{2}$ &
0 &
0 &
0.5 &
0.5\tabularnewline
\hline 
$\beta_{0}$ &
2 &
2 &
2.5 &
2.5\tabularnewline
\hline 
$\delta_{1}$ &
1.22 &
1.26 &
1 &
1.05\tabularnewline
\hline 
$\delta_{2}$ &
0 &
0 &
0.5 &
0.5\tabularnewline
\hline 
$\phi_{0}$ &
0.1 &
0.1 &
0.1 &
0.1\tabularnewline
\hline 
$\phi_{1}$ &
0.5 &
0.5 &
0.5 &
0.5\tabularnewline
\hline 
$\phi_{2}$ &
0 &
0.1 &
0 &
0.1\tabularnewline
\hline 
\end{tabular}
\par\end{centering}
\uline{Notes:} We specified $a_{0}\left(w\right)=\alpha_{0}+\gamma_{1}\left(W-\mathbb{E}\left[W\right]\right)+\gamma_{2}\left(W^{2}-\mathbb{E}\left[W^{2}\right]\right)$
and $b_{0}\left(w\right)=\beta_{0}+\delta_{1}\left(W-\mathbb{E}\left[W\right]\right)+\delta_{2}\left(W^{2}-\mathbb{E}\left[W^{2}\right]\right)$
analogous to the formulation given in Assumption \ref{ass: clp_coefficients}.
Each of the four designs are calibrated such that $\sqrt{\mathcal{I}\left(\beta_{0}\right)^{-1}/N}=0.05$
when $N=1,000$. 
\end{table}

We evaluate the performance of three estimators. First we consider
a simple ``Oaxaca-Blinder'' type estimator. Specifically we estimate
$\beta_{0}$ by the coefficient on $X$ in the least squares fit of
$Y$ onto a constant, $W-\hat{\mu}_{W}$, $\left(W-\hat{\mu}_{W}\right)\times X$
and $X$. As in \citet{Kline_EL14} we appropriately account for the
effect of estimating $\mu_{W}$ when constructing standard errors
and confidence intervals. This estimator is consistent for the true
average partial effect in both designs 1 and 2. It is also, since
$U$ is Gaussian and homoscedastic, efficient in these two designs.
Efficiency is in the semiparametric model which, \emph{in addition
to} Assumptions \ref{ass: random_sampling} and \ref{ass: overlap},
maintains Assumption \ref{ass: clp_coefficients}. The variance of
the Oaxaca-Blinder estimate therefore lies (weakly) below the bound
given by Theorem \ref{thm: SEB} in these two designs. In designs
3 and 4, where $a_{0}\left(W\right)$ and $b_{0}\left(W\right)$ are
quadratic in $W$, the ``Oaxaca-Blinder'' estimator is inconsistent.

The second estimator is the generalized inverse probability weighting
(GIPW) one due to \citet{Wooldridge_CWP04,Wooldridge_EACSPDBook10}.
Our implementation tracks our analysis in Section \ref{sec: estimation}.
For estimation we correctly assume that the conditional distribution
of $X$ given $W$ is Poisson, but set the parameter to $\exp\left(k\left(W\right)'\phi\right)$
with $k\left(W\right)=\left(1,W\right)'$. This is correct in designs
1 and 3, but not 2 and 4. Hence the GIPW estimate of $\beta_{0}$
is consistent in designs 1 and 3, but not 2 and 4. The GIPW is never
efficient. Standard errors are constructed using the sample analog
of the influence function given in (\ref{eq: Wooldridge_Influence})
above.

Finally we consider the properties of our locally efficient, doubly
robust estimator. Implementing this procedure requires assumptions
on both the CLP and the generalized propensity score. We make the
same assumptions used to implement the Oaxaca-Blinder and GIPW procedures.
Consequently this last estimator is efficient \textendash{} in the
sense of Theorem \ref{thm: SEB} \textendash{} in design 1 and consistent
in designs 1, 2 and 3. Like all the estimators it is inconsistent
in design 4. We construct standard errors using the (sample analog)
of the influence function given in Theorem \ref{thm: Large-Sample};
consequently our intervals are conservative in design 2 (where our
propensity score model is misspecified).\footnote{We use Python 3.6 to conduct our experiments. Replication code is
available in the supplemental materials. }

Each of the four designs are calibrated such that $\sqrt{\mathcal{I}\left(\beta_{0}\right)^{-1}/N}=0.05$
($0.025$) when $N=1,000$ ($4,000$). In an asymptotic sense inference
on $\beta_{0}$ is equally hard across all the designs considered.
We focus on the $N=1,000$ experiments in our discussion (the quality
of the asymptotic approximations predictably improve in the larger
sample).

Results from the four designs are reported in Table \ref{tab: Simulation-Results}.
As expected our DR estimator is median unbiased across Designs 1,
2 and 3. In contrast the Oaxaca-Blinder estimator only performs acceptably
in designs 1 and 2, and the GIPW estimator in design 1 and 3. In designs
1 and 2 the variability of the DR estimator is nearly as small as
that of the Oaxaca-Blinder one. Neither the DR, nor the GIPW, estimators
are expected to be efficient in design 3 but, interestingly, GIPW
is more efficient than DR in this case. In design 1, where the DR
estimator is locally efficient, its standard deviation is substantially
smaller than that of the GIPW estimator (as expected).

\begin{landscape}

\begin{table}
\caption{Simulation Results\label{tab: Simulation-Results}}

\begin{centering}
\begin{tabular}{|c|c|c|c|c|c|c|c|c|}
\hline 
 &
\multicolumn{4}{c|}{\textbf{Panel A}, $N=1,000$} &
\multicolumn{4}{c|}{\textbf{Panel B}, $N=4,000$}\tabularnewline
\hline 
 &
Bias &
Std. Dev. &
Std. Err. &
Coverage &
Bias &
Std. Dev. &
Std. Err. &
Coverage\tabularnewline
\hline 
\hline 
\textbf{Design 1} &
 &
 &
 &
 &
 &
 &
 &
\tabularnewline
\hline 
Oaxaca-Blinder &
-0.0003 &
0.0500 &
0.0496 &
0.9480 &
0.0006 &
0.0252 &
0.0248 &
0.9438\tabularnewline
\hline 
GIPW &
-0.0008 &
0.0853 &
0.0809 &
0.9438 &
0.0016 &
0.0429 &
0.0419 &
0.9494\tabularnewline
\hline 
DR &
0.0001 &
0.0507 &
0.0499 &
0.9450 &
0.0009 &
0.0254 &
0.0250 &
0.9448\tabularnewline
\hline 
\textbf{Design 2} &
 &
 &
 &
 &
 &
 &
 &
\tabularnewline
\hline 
Oaxaca-Blinder &
0.0009 &
0.0504 &
0.0497 &
0.9442 &
0.0000 &
0.0251 &
0.0248 &
0.9456\tabularnewline
\hline 
GIPW &
-0.2597 &
0.1331 &
0.1198 &
0.4634 &
-0.2613 &
0.0710 &
0.0666 &
0.0258\tabularnewline
\hline 
DR &
0.0014 &
0.0518 &
0.0561 &
0.9620 &
-0.0001 &
0.0258 &
0.0283 &
0.9694\tabularnewline
\hline 
\textbf{Design 3} &
 &
 &
 &
 &
 &
 &
 &
\tabularnewline
\hline 
Oaxaca-Blinder &
-0.3268 &
0.0993 &
0.0899 &
0.0772 &
-0.3319 &
0.0531 &
0.0481 &
0.0002\tabularnewline
\hline 
GIPW &
-0.0018 &
0.0830 &
0.0794 &
0.9436 &
-0.0005 &
0.0428 &
0.0413 &
0.9452\tabularnewline
\hline 
DR &
-0.0113 &
0.1099 &
0.0981 &
0.9276 &
-0.0036 &
0.0570 &
0.0529 &
0.9368\tabularnewline
\hline 
\textbf{Design 4} &
 &
 &
 &
 &
 &
 &
 &
\tabularnewline
\hline 
Oaxaca-Blinder &
-0.2010 &
0.1571 &
0.1087 &
0.5148 &
-0.2391 &
0.1255 &
0.0674 &
0.0168\tabularnewline
\hline 
GIPW &
0.2717 &
0.1897 &
0.1216 &
0.4006 &
0.3052 &
0.1335 &
0.0748 &
0.0034\tabularnewline
\hline 
DR &
0.4123 &
0.2035 &
0.1687 &
0.3076 &
0.4362 &
0.1058 &
0.1013 &
0.0986\tabularnewline
\hline 
\hline 
$\sqrt{\mathcal{I}\left(\beta_{0}\right)^{-1}/N}$ &
\multicolumn{4}{c|}{0.05} &
\multicolumn{4}{c|}{0.0250}\tabularnewline
\hline 
\end{tabular}
\par\end{centering}
\uline{Notes:} The bias column reports median bias across all $B=5,000$
simulations. The Std. Dev. column reports the standard deviation of
the point estimates across these simulations, Std. Err. the median
estimated standard error, and Coverage the actual coverage of a nominal
95 percent confidence interval (constructed using the estimated point
estimate and standard error in the normal way). The standard error
associated with a Monte Carlo coverage estimate is $\sqrt{\alpha\left(1-\alpha\right)/B}$.
With $B=5,000$ simulations and $\alpha=0.05$, this results in a
standard error of approximately 0.003 or a 95 percent confidence interval
of $[0.944,0.956]$.
\end{table}

\end{landscape}

Overall the simulation results track our theoretical predictions remarkably
closely. Of course exploring the performance of these estimators in
the context of real world empirical applications and other, more realistic,
simulation designs would be of interest.

\section{Conclusion}

We have introduced a locally efficient, doubly robust, semiparametric
method of estimating averages of conditional linear predictor coefficients.
Our estimand, and semiparametric efficiency bound, specialize to familiar
counterparts found in the program evaluation literature \citep[e.g.][]{Hahn_EM98,Cattaneo_JOE10}.
While encompassing well-known program evaluation settings, our framework
allows for (semiparametric) covariate adjustment in many other settings
as well (including ones with few extant alternative methods of such
adjustment).

Researchers interested in estimating the average treatment effect
(ATE) associated with a binary treatment can apply our methods. While
we believe the precise form of our procedure is new even to this familiar
setting, it is a variant of the class of augmented inverse probability
weighting (AIPW) estimators introduced by \citet{Robins_Rotnitzky_Zhao_JASA94}
in the missing data context over 20 years ago. The real attraction
of Algorithm \ref{alg: Estimation}, and the corresponding Theorems
\ref{thm: Large-Sample} and \ref{thm:DoublyRobust} (as well as Proposition
\ref{prop: near_efficiency}), is that they apply to models beyond
the ``classic'' program evaluation one of \citet{Rosenbaum_Rubin_BM83}.
Multiple, mutually exclusive treatments, as in \citet{Imbens_BM00}
and \citet{Cattaneo_JOE10} are easily handled as a special case.
Similarly, maintaining a linear, but heterogeneous, potential response
function structure, Algorithm \ref{alg: Estimation} recovers average
partial effects (APE) for continuous treatments, multiple non-exclusive
treatments, mixtures of binary, discrete and continuous treatments
and so on. A weighted average derivative interpretation of our estimand
is also available for settings where the linear potential response
function structure may not hold (Proposition \ref{prop: wgt_der}).

We also wish to emphasize that averages of conditional linear predictor
coefficients represent a natural, but substantial, generalization
of linear predictor coefficients as estimated by the method of least
squares. Hence Algorithm \ref{alg: Estimation} also provides a method
of flexible covariate adjustment that may be of independent interest
even in settings where formal causal inference is not warranted; similar
to how least squares is sometimes used for descriptive purposes.

Our work leaves several questions unanswered. First, although the
flexible parametric modeling embodied in Assumptions \ref{ass: genalize_p_score}
and \ref{ass: clp_coefficients} closely mirrors empirical practice,
it would be useful to development methods that leave the generalized
propensity score and CLP coefficients non-parametric. It seems likely
that results from the binary and multiple treatments case could be
extended to apply here \citep[e.g.,][]{Hirano_et_al_EM03,Cattaneo_JOE10,Belloni_et_al_ReStud14}. 

In other work we have shown that first order equivalent estimators
may have appreciably different higher order properties in program
evaluation settings \citep{Graham_Pinto_Egel_ReStud12}. We expect
that other locally efficient, doubly robust approaches to estimation
for the class of problems considered in this paper are feasible. These
approaches may exhibit superior or inferior higher order bias.

Third, maintaining the correlated random coefficient structure, different
notions of conditional exogeneity will imply different semiparametric
efficiency bounds (when linearity is restrictive). Our decision to
work with a weak notion of exogeneity maintains a connection with
conditional linear predictors. If a researcher was comfortable with
the correlated random coefficient structure, then it would generally
be possible to construct more efficient estimates of $\beta_{0}=\mathbb{E}\left[B\right]$
if she was willing to assume, for example, that $\left.\left(A,B'\right)'\perp X\right|W=w$
for all $w\in\mathbb{W}.$ Such estimators would likely be quite complicated
and may have poor finite sample properties.

\subsection*{\protect\pagebreak{}}

\appendix

\section{\label{appendix: proofs}Proofs}

This appendix contains proofs of the results contained in the main
paper. All notation is as defined in the main text unless explicitly
noted otherwise. Equation numbering continues in sequence with that
established in the main text. 

\subsection*{Proof of Theorem \ref{thm: SEB} (Semiparametric efficiency bound)}

In calculating the efficiency bound for $\beta_{0}$ in the semiparametric
regression model defined by Definition \ref{def: clp} and Assumptions
\ref{ass: random_sampling} and \ref{ass: overlap} of the main text,
we follow the approach outlined by \citet[Section 3]{Newey_JAE90}.
First, we characterize the model's tangent space. Second, we demonstrate
pathwise differentiability of $\beta_{0}$. The efficient influence
function for $\beta_{0}$ equals the projection of this derivative
onto the tangent space. In the present case the pathwise derivative
lies in the tangent space and hence coincides with the required projection.
The result then follows from an application of Theorem 3.1 in \citet{Newey_JAE90}.

\subsubsection*{Step 1: Characterization of the Model Tangent Space:}

The joint density function for $z=\left(w,x,y\right)$ is given by
\[
f_{0}\left(w,x,y\right)=f_{0}\left(\left.x,y\right|w\right)f_{0}\left(w\right),
\]
where $f_{0}\left(\left.x,y\right|w\right)$ denotes the conditional
density/mass of $\left(X=x,Y=y\right)$ given $W=w$ and $f_{0}\left(w\right)$
is the marginal density/mass of $W=w$.

Consider a regular parametric submodel indexed by $\eta$ with $f\left(w,x,y;\eta\right)=f_{0}\left(w,x,y\right)$
at $\eta=\eta_{0}$. The submodel joint density equals
\[
f\left(w,x,y;\eta\right)=f\left(\left.x,y\right|w;\eta\right)f\left(w;\eta\right),
\]
with a corresponding score vector of
\begin{equation}
s_{\eta}\left(w,x,y;\eta\right)=s_{\eta}\left(\left.x,y\right|w;\eta\right)+t_{\eta}\left(w;\eta\right)\label{eq: submodel_score}
\end{equation}
where
\[
s_{\eta}\left(w,x,y;\eta\right)=\nabla_{\eta}f\left(w,x,y;\eta\right),\thinspace\thinspace s_{\eta}\left(\left.x,y\right|w;\eta\right)=\nabla_{\eta}f\left(\left.x,y\right|w;\eta\right),\thinspace\thinspace t_{\eta}\left(w;\eta\right)=\nabla_{\eta}f\left(w;\eta\right).
\]
By the usual (conditional) mean zero property of scores we have that
\begin{equation}
\mathbb{E}\left[\left.s_{\eta}\left(\left.X,Y\right|W\right)\right|W\right]=\mathbb{E}\left[t_{\eta}\left(W\right)\right]=0,\label{eq: submodel_score_restrictions}
\end{equation}
where the suppression of $\eta$ in a function indicates that it is
evaluated at its population value (e.g., $t_{\eta}\left(w\right)=t_{\eta}\left(w;\eta_{0}\right)$).

The model tangent set is the closed linear span of the set of all
such scores. From (\ref{eq: submodel_score}) and (\ref{eq: submodel_score_restrictions})
this set evidently equals
\[
\mathcal{T}=\left\{ s\left(\left.x,y\right|w\right)+t\left(w\right)\right\} 
\]
where $s\left(\left.x,y\right|w\right)$ and $t\left(w\right)$ satisfy
the (conditional) moment restrictions
\[
\mathbb{E}\left[\left.s\left(\left.X,Y\right|W\right)\right|W\right]=\mathbb{E}\left[t\left(W\right)\right]=0,
\]
and also have finite variance.

\subsubsection*{Step 2: Demonstration of pathwise differentiability:}

Under the parametric submodel, $\beta\left(\eta\right)$ is identified
by 
\begin{equation}
\beta\left(\eta\right)=\int b\left(w;\eta\right)f\left(w;\eta\right)\mathrm{d}w,\label{eq: beta(eta)}
\end{equation}
where $b\left(w;\eta\right)$ satisfies the conditional moment restriction
\begin{equation}
\int\int\left(\begin{array}{c}
1\\
x
\end{array}\right)\left(y-a\left(w;\eta\right)-x'b\left(w;\eta\right)\right)f\left(\left.x,y\right|w;\eta\right)\mathrm{d}x\mathrm{d}y=0.\label{eq: CLP(eta)}
\end{equation}

Differentiating (\ref{eq: beta(eta)}) under the integral and evaluating
at $\eta=\eta_{0}$ gives
\begin{equation}
\frac{\partial\beta\left(\eta_{0}\right)}{\partial\eta'}=\mathbb{E}\left[\frac{\partial b\left(W;\eta_{0}\right)}{\partial\eta'}\right]+\mathbb{E}\left[b\left(W;\eta_{0}\right)t_{\eta}\left(W;\eta_{0}\right)\right].\label{eq: pathwise_1}
\end{equation}
We can derive a close-form expression for $\frac{\partial b\left(w;\eta_{0}\right)}{\partial\eta'}$
in (\ref{eq: pathwise_1}) by differentiating (\ref{eq: CLP(eta)})
with respect to $\eta$ (and evaluating at $\eta=\eta_{0}$): 
\begin{eqnarray*}
-\int\int\left(\begin{array}{c}
1\\
x
\end{array}\right)\frac{\partial a\left(w;\eta_{0}\right)}{\partial\eta'}f\left(\left.x,y\right|w;\eta_{0}\right)\mathrm{d}x\mathrm{d}y-\int\int\left(\begin{array}{c}
x'\\
xx'
\end{array}\right)\frac{\partial b\left(w;\eta_{0}\right)}{\partial\eta'}f\left(\left.x,y\right|w;\eta_{0}\right)\mathrm{d}x\mathrm{d}y\\
+\int\int\left(\begin{array}{c}
1\\
x
\end{array}\right)\left(y-a\left(w;\eta\right)-x'b\left(w;\eta_{0}\right)\right)s_{\eta}\left(\left.x,y\right|w;\eta_{0}\right)f\left(\left.x,y\right|w;\eta_{0}\right)\mathrm{d}x\mathrm{d}y & = & 0
\end{eqnarray*}
Using the matrix inverse
\[
\mathbb{E}\left[\left.\begin{array}{cc}
1 & X'\\
X & XX'
\end{array}\right|W=w\right]^{-1}=\left(\begin{array}{cc}
1+e_{0}\left(w\right)'v_{0}\left(w\right)^{-1} & e_{0}\left(w\right)-e_{0}\left(w\right)'v_{0}\left(w\right)^{-1}\\
-v_{0}\left(w\right)^{-1}e_{0}\left(w\right) & v_{0}\left(w\right)^{-1}
\end{array}\right)
\]
we solve to get
\begin{eqnarray*}
\left(\begin{array}{c}
\frac{\partial a\left(w;\eta_{0}\right)}{\partial\eta'}\\
\frac{\partial b\left(w;\eta_{0}\right)}{\partial\eta'}
\end{array}\right) & = & \left(\begin{array}{cc}
1+e_{0}\left(w\right)'v_{0}\left(w\right)^{-1} & e_{0}\left(w\right)-e_{0}\left(w\right)'v_{0}\left(w\right)^{-1}\\
-v_{0}\left(w\right)^{-1}e_{0}\left(w\right) & v_{0}\left(w\right)^{-1}
\end{array}\right)\\
 &  & \times\mathbb{E}\left[\left.\left(\begin{array}{c}
Y-a\left(W;\eta_{0}\right)-X'b\left(W;\eta_{0}\right)\\
X\left(Y-a\left(W;\eta_{0}\right)-X'b\left(W;\eta_{0}\right)\right)
\end{array}\right)s_{\eta}\left(\left.X,Y\right|W;\eta_{0}\right)\right|W=w\right].
\end{eqnarray*}
Evaluating the second row of this expression gives
\begin{equation}
\frac{\partial b\left(w;\eta_{0}\right)}{\partial\eta'}=\mathbb{E}\left[\left.v_{0}\left(W\right)^{-1}\left(X-e_{0}\left(W\right)\right)\left(Y-a_{0}\left(W\right)-X'b_{0}\left(W\right)\right)s_{\eta}\left(\left.X,Y\right|W\right)\right|W=w\right],\label{eq: d_b(W)/d_eta}
\end{equation}
which, after substituting into (\ref{eq: pathwise_1}), yields
\begin{eqnarray}
\frac{\partial\beta\left(\eta_{0}\right)}{\partial\eta'} & = & \mathbb{E}\left[v_{0}\left(W\right)^{-1}\left(X-e_{0}\left(W\right)\right)\left(Y-a_{0}\left(W\right)-X'b_{0}\left(W\right)\right)s_{\eta}\left(\left.X,Y\right|W\right)\right]\nonumber \\
 &  & +\mathbb{E}\left[b_{0}\left(W\right)t_{\eta}\left(W\right)\right].\label{eq: pathwise_2}
\end{eqnarray}

To demonstrate pathwise differentiability of $\beta$, we require
$F\left(W,X,Y\right)$ such that
\begin{equation}
\frac{\partial\beta\left(\eta_{0}\right)}{\partial\eta'}=\mathbb{E}\left[F\left(W,X,Y\right)s_{\eta}\left(W,X,Y\right)'\right].\label{eq: pathwise_dif_condition}
\end{equation}

Setting $F\left(W,X,Y\right)$ equal to $\mathrm{\psi}_{\beta}^{\mathrm{eff}}\left(Z,\beta_{0},g_{0}\left(W\right),h_{0}\left(W\right)\right)$,
as defined in (\ref{eq: EfficientInfluenceFunction_CRC}) of the main
text, we get $\mathbb{E}\left[F\left(W,X,Y\right)s_{\eta}\left(W,X,Y\right)'\right]$
equal to (\ref{eq: pathwise_2}) since, by Lemma 4.1 of \citet{Wooldridge_JE99},
\begin{eqnarray*}
\mathbb{E}\left[\left.\left(X-e_{0}\left(W\right)\right)\left(Y-a_{0}\left(W\right)-X'b_{0}\left(W\right)\right)\right|W\right] & = & 0
\end{eqnarray*}
and iterated expectations (and the conditional mean zero property
of the score $s_{\eta}\left(\left.X,Y\right|W\right)$) further implies
that $\mathbb{E}\left[\left(b_{0}\left(W\right)-\beta_{0}\right)s_{\eta}\left(\left.X,Y\right|W\right)\right]=0$. 

\subsubsection*{Step 3: Verification that the conjectured influence function equals
the required projection:}

Observe that $\mathrm{\psi}_{\beta}^{\mathrm{eff}}\left(Z,\beta_{0},g\left(W\right),h\left(W\right)\right)$
lies in the model tangent space. Its first term is conditionally mean
zero given $W$ and hence plays the role of $s\left(\left.X,Y\right|W\right)$
. Its second term is a mean zero function of $W$ alone and hence
plays the role of $t\left(W\right)$. Since $\mathrm{\psi}_{\beta}^{\mathrm{eff}}\left(Z,\beta_{0},g_{0}\left(W\right),h_{0}\left(W\right)\right)\in\mathcal{T}$,
its projection onto $\mathcal{T}$ equals itself. Since equation (9)
of \citet[p. 106]{Newey_JAE90} is satisfied the result follows from
his Theorem 3.1.

\subsection*{Proof of Theorem \ref{thm: Large-Sample} (Large sample properties
of $\hat{\beta}$ )}

Recall that $\lambda=\left(\alpha,\gamma',\delta'\right)'$ and 
\begin{align*}
\underset{\left(1+J+JK\right)\times1}{R} & =\left(1,\left(W-\mu_{W}\right)',\left(\left(W-\mu_{W}\right)\otimes X\right)'\right)'.
\end{align*}
In what follows we let $\lambda_{*}$ denote value of $\lambda$ which
solves the just-identified population moments (\ref{eq: m1_p1}),
(\ref{eq: m2_p1}) and (\ref{eq: m3_p1}). If Assumption \ref{ass: clp_coefficients}
additionally holds in the sampled population, then we use $\lambda_{0}$
to denote the population value of $\lambda$. In this case $\lambda_{0}$
correctly specifies the form of the CLP of $Y$ given $X$ conditional
on $W$. 

In the Supplemental Web Appendix we show, without maintaining Assumption
\ref{ass: clp_coefficients}, that
\begin{align}
\lambda_{*}= & \mathbb{E}\left[RR'\right]^{-1}\mathbb{E}\left[R\left(Y-X'\beta_{0}\right)\right]\label{eq: lambda_BLP}\\
= & \mathbb{E}\left[RR'\right]^{-1}\mathbb{E}\left[R\left(a_{0}\left(W\right)+X'\left(b_{0}\left(W\right)-\beta_{0}\right)\right)\right].\nonumber 
\end{align}
Equation (\ref{eq: lambda_BLP}) implies that $R'\lambda_{*}$ is
the mean squared error minimizing linear predictor of $a_{0}\left(W\right)+X'\left(b_{0}\left(W\right)-\beta_{0}\right)$
given $R$. This interpretation of $\lambda_{*}$ is all that is required
for the first part of Theorem \ref{thm: Large-Sample}.

We will also use the notation
\begin{align*}
U_{0}= & \left(Y-R'\lambda_{0}-X'\beta_{0}\right)
\end{align*}
and
\[
U_{*}=\left(Y-R'\lambda_{*}-X'\beta_{0}\right).
\]
Note that under Assumption \ref{ass: clp_coefficients} $U_{0}$ equals
a \emph{conditional} linear prediction error. However when Assumption
\ref{ass: clp_coefficients} does not hold an implication of (\ref{eq: lambda_BLP})
is that $U_{*}$ is still an \emph{unconditional} linear predictor
error.

We also use the shorthand $e_{0}\left(W\right)=e\left(W;\phi_{0}\right)$
and $v_{0}\left(W\right)=v\left(W;\phi_{0}\right)$ in order to simplify
some of the expressions presented below. Finally we use $\theta_{0}$
to denote \emph{both} $\left(\phi_{0}',\mu_{W}',\lambda_{*}',\beta_{0}'\right)'$
and $\left(\phi_{0}',\mu_{W}',\lambda_{0}',\beta_{0}'\right)'$, with
the relevant case made clear by the context.

Next define the $\left(1+J+JK\right)\times J$ and $K\times J$ matrices
\begin{align}
\underset{\left(1+J+JK\right)\times J}{B_{1}} & =\mathbb{E}\left[R\left\{ \gamma_{*}+\left(I_{J}\otimes X\right)^{^{\prime}}\delta_{*}\right\} '\right]-\mathbb{E}\left[\begin{array}{c}
0\\
0\\
\left(I_{J}\otimes X\right)U_{*}
\end{array}\right]\label{eq: B1_matrix}\\
\underset{K\times J}{B_{2}} & =\mathbb{E}\left[\left(v_{0}\left(W\right)^{-1}\left(X-e_{0}\left(W\right)\right)\right)\left\{ \gamma_{*}+\left(I_{J}\otimes X\right)^{^{\prime}}\delta_{*}\right\} '\right].\label{eq: B2_matrix}
\end{align}
Using this notation we can write the $\left(L+J+1+J+JK+K\right)\times\left(L+J+1+J+JK+K\right)$
Jacobian matrix of the moment vector as
\[
M=-\left(\begin{array}{cccc}
-\mathbb{H}\left(\phi_{0}\right) & 0 & 0 & 0\\
0 & I_{J} & 0 & 0\\
0 & -B_{1} & \mathbb{E}\left[RR'\right] & \mathbb{E}\left[RX'\right]\\
\mathbb{E}\left[v_{0}\left(W\right)^{-1}\left(X-e_{0}\left(W\right)\right)U_{*}\mathbb{S}'\right] & -B_{2} & 0 & I_{K}
\end{array}\right),
\]
with $\mathbb{H}\left(\phi_{0}\right)$ equal to the $L\times L$
expected Hessian matrix associated with the generalized propensity
score log-likelihood. The inverse Jacobian is therefore
\begin{align}
M^{-1} & =-\left(\begin{array}{c}
-\mathbb{H}\left(\phi_{0}\right)^{-1}\\
0\\
-\mathbb{E}\left[RR'\right]^{-1}\mathbb{E}\left[RX'\right]\mathbb{E}\left[v_{0}\left(W\right)^{-1}\left(X-e_{0}\left(W\right)\right)U_{*}\mathbb{S}'\right]\mathbb{H}\left(\phi_{0}\right)^{-1}\\
\mathbb{E}\left[v_{0}\left(W\right)^{-1}\left(X-e_{0}\left(W\right)\right)U_{*}\mathbb{S}'\right]\mathbb{H}\left(\phi_{0}\right)^{-1}
\end{array}\right.\nonumber \\
 & \left.\begin{array}{ccc}
0 & 0 & 0\\
I_{J} & 0 & 0\\
\mathbb{E}\left[RR'\right]^{-1}\left(B_{1}-B_{2}\mathbb{E}\left[RX'\right]\right) & \mathbb{E}\left[RR'\right]^{-1} & -\mathbb{E}\left[RR'\right]^{-1}\mathbb{E}\left[RX'\right]\\
B_{2} & 0 & I_{K}
\end{array}\right).\label{eq: inverse_Jacobian}
\end{align}

Under Assumption \ref{ass: genalize_p_score} a key observation is
that the expected value of (\ref{eq: B2_matrix}) equals
\begin{align*}
\mathbb{E}\left[\left(v_{0}\left(W\right)^{-1}\left(X-e_{0}\left(W\right)\right)\right)\left\{ \gamma_{*}+\left(I_{J}\otimes X\right)^{^{\prime}}\delta_{*}\right\} '\right] & =\mathbb{E}\left[\left(v_{0}\left(W\right)^{-1}\left(X-e_{0}\left(W\right)\right)\right)\gamma_{*}'\right]\\
 & +\mathbb{E}\left[\left(v_{0}\left(W\right)^{-1}\left(X-e_{0}\left(W\right)\right)\right)\delta_{*}'\left(I_{J}\otimes X\right)\right]\\
= & 0+\mathbb{E}\left[\left(v_{0}\left(W\right)^{-1}\left(X-e_{0}\left(W\right)\right)\right)\right.\\
 & \left.\times\left(\begin{array}{ccc}
X'\delta_{1*} & \cdots & X'\delta_{J*}\end{array}\right)\right]\\
= & \left(\begin{array}{ccc}
\delta_{1*} & \cdots & \delta_{J*}\end{array}\right)=\Delta_{*}.
\end{align*}

Using this last equality, as well as the fact that under Assumption
\ref{ass: genalize_p_score} we have $\mathbb{H}\left(\phi_{0}\right)=-\mathbb{E}\left[\mathbb{S}\mathbb{S}'\right]$,
implies that the last $K$ rows of $-M^{-1}\frac{1}{\sqrt{N}}\sum_{i=1}^{N}m\left(Z_{i},\theta_{0}\right)+o_{p}\left(1\right)$
equal, after some manipulation,
\begin{align}
\sqrt{N}\left(\hat{\beta}-\beta_{0}\right)= & \frac{1}{\sqrt{N}}\sum_{i=1}^{N}\left\{ v_{0}\left(W_{i}\right)^{-1}\left(X_{i}-e_{0}\left(W_{i}\right)\right)U_{*i}\right.\nonumber \\
 & -\mathbb{E}\left[v_{0}\left(W\right)^{-1}\left(X-e_{0}\left(W\right)\right)U_{*}\mathbb{S}'\right]\mathbb{E}\left[\mathbb{S}\mathbb{S}'\right]^{-1}\mathbb{S}_{i}\nonumber \\
 & \left.+\Delta_{*}\left(W_{i}-\mu_{W}\right)\right\} +o_{p}\left(1\right).\label{eq: beta_hat_general_influence_a}
\end{align}

Next observe that we may decompose $U_{*}$ as
\begin{align*}
U_{*}= & Y-R'\lambda_{*}-X'\beta_{0}\\
= & Y-a_{0}\left(W\right)-X'b_{0}\left(W\right)\\
 & +\left\{ a_{0}\left(W\right)+X'\left(b_{0}\left(W\right)-\beta_{0}\right)-R'\lambda_{*}\right\} \\
= & U_{0}+\epsilon.
\end{align*}
Since $\mathbb{E}\left[U_{*}W\right]=0$ by the properties of linear
predictors, $\mathbb{E}\left[\left.U_{0}\right|W\right]=0$ by the
properties of \emph{conditional }linear predictors, and $U_{*}=U_{0}+\epsilon$,
we have that $\mathbb{E}\left[\epsilon W\right]=0$. Defining $\tilde{\epsilon}=v_{0}\left(W\right)^{-1}\left(X-e_{0}\left(W\right)\right)\epsilon$
we can re-write \ref{eq: beta_hat_general_influence_a} as
\begin{align}
\sqrt{N}\left(\hat{\beta}-\beta_{0}\right) & =\frac{1}{\sqrt{N}}\sum_{i=1}^{N}\left\{ v_{0}\left(W_{i}\right)^{-1}\left(X_{i}-e_{0}\left(W_{i}\right)\right)U_{0i}\right.\nonumber \\
 & \left.+\left(\tilde{\epsilon}_{i}-\Pi_{\tilde{\epsilon}\mathbb{S}}\mathbb{S}_{i}\right)+\Delta_{*}\left(W_{i}-\mu_{W}\right)\right\} +o_{p}\left(1\right)\label{eq: beta_hat_general_influence_b}
\end{align}
where $\Pi_{\tilde{\epsilon}\mathbb{S}}=\mathbb{E}\left[\tilde{\epsilon}_{i}\mathbb{S}'\right]\mathbb{E}\left[\mathbb{S}\mathbb{S}'\right]^{-1}$.
This gives the first implication of the Theorem. The second implication
follows from the fact that $\epsilon=0$ and $\Delta_{*}\mathbb{V}\left(W\right)\Delta_{*}'=\mathbb{V}\left(b_{0}\left(W\right)\right)$
under Assumption \ref{ass: clp_coefficients}.

\subsection*{Proof of Proposition \ref{prop: near_efficiency} (Near global semiparametric
efficiency)}

Let $\mathbf{A}$ be an $m\times n$ matrix with $\left\Vert \mathbf{A}\right\Vert _{F}=\mathrm{Tr}\left(\mathbf{A}'\mathbf{A}\right)^{1/2}$
denoting the Frobenius matrix norm, $\left\Vert \mathbf{A}\right\Vert _{2}$
the spectral norm and recall that $\left\Vert \mathbf{A}\right\Vert _{2}\leq\left\Vert \mathbf{A}\right\Vert _{F}$.
Let $\mathbf{a}$ be an $n\times1$ vector with Euclidean norm $\left\Vert \mathbf{a}\right\Vert =\left(\mathbf{a}'\mathbf{a}\right)^{1/2}$.
We make use of several matrix and probability inequalities in what
follows. These are drawn from \citet[Appendices A \& B]{Hansen_Book18}
unless stated otherwise.

Let $t$ be a non-zero column vector. The difference in the asymptotic
variance of the estimate of the linear combination $t'\beta_{0}$
based upon $R^{\left(J\right)}$ and a corresponding semiparametrically
efficient estimate is
\begin{align}
t'\mathcal{I}^{\left(J\right)}\left(\beta_{0}\right)^{-1}t-t'\mathcal{I}\left(\beta_{0}\right)^{-1}t= & t'\Delta_{*}^{\left(J\right)}\mathbb{V}\left(k^{\left(J\right)}\left(W\right)\right)\left(\Delta_{*}^{\left(J\right)}\right)'t-t'\mathbb{V}\left(b_{0}\left(W\right)\right)t\nonumber \\
 & +t'\mathbb{E}\left[\left(\tilde{\epsilon}^{\left(J\right)}-\Pi_{\tilde{\epsilon}\mathbb{S}}^{\left(J\right)}\mathbb{S}\right)\left(\tilde{\epsilon}^{\left(J\right)}-\Pi_{\tilde{\epsilon}\mathbb{S}}^{\left(J\right)}\mathbb{S}\right)'\right]t\nonumber \\
\geq & 0.\label{eq: var_dif_1}
\end{align}
We seek to show that this variance difference is also bounded above
by something that can be made arbitrarily close to zero.

To start observe that, after some manipulation (see the Supplemental
Web Appendix) we can show that
\begin{align}
\mathbb{V}\left(b_{0}\left(W\right)+\Delta_{*}^{\left(J\right)}\left(k^{\left(J\right)}\left(W\right)-\mu^{\left(J\right)}\right)-\beta_{0}\right)= & \Delta_{*}^{\left(J\right)}\mathbb{V}\left(k^{\left(J\right)}\left(W\right)\right)\left(\Delta_{*}^{\left(J\right)}\right)'-\mathbb{V}\left(b_{0}\left(W\right)\right)\nonumber \\
 & -2\mathbb{E}\left[\left(b_{0}\left(W\right)-\beta_{0}\right)\right.\nonumber \\
 & \left.\times\left\{ b_{0}\left(W\right)+\Delta_{*}^{\left(J\right)}\left(k^{\left(J\right)}\left(W\right)-\mu^{\left(J\right)}\right)-\beta_{0}\right\} '\right]\label{eq: rewritten_key_variance}
\end{align}
Using (\ref{eq: rewritten_key_variance}) we can rewrite $t'\mathcal{I}^{\left(J\right)}\left(\beta_{0}\right)^{-1}t-t'\mathcal{I}\left(\beta_{0}\right)^{-1}t$
as 
\begin{align}
 & t'\mathbb{V}\left(b_{0}\left(W\right)+\Delta_{*}^{\left(J\right)}\left(k^{\left(J\right)}\left(W\right)-\mu^{\left(J\right)}\right)-\beta_{0}\right)t\nonumber \\
 & +2t'\mathbb{E}\left[\left(b_{0}\left(W\right)-\beta_{0}\right)\left\{ b_{0}\left(W\right)+\Delta_{*}^{\left(J\right)}\left(k^{\left(J\right)}\left(W\right)-\mu^{\left(J\right)}\right)-\beta_{0}\right\} '\right]t\nonumber \\
 & +t'\mathbb{E}\left[\left(\tilde{\epsilon}^{\left(J\right)}-\Pi_{\tilde{\epsilon}\mathbb{S}}^{\left(J\right)}\mathbb{S}\right)\left(\tilde{\epsilon}^{\left(J\right)}-\Pi_{\tilde{\epsilon}\mathbb{S}}^{\left(J\right)}\mathbb{S}\right)'\right]t.\label{eq: var_dif_2}
\end{align}
Consider the first term in (\ref{eq: var_dif_2}). The Quadratic Inequality
(QI), Expectation Inequality (EI), and completeness of the sequence
$\left\{ k_{j}\left(W\right)\right\} _{j=1}^{\infty}$ (see equation
(\ref{eq: L2-complete})) give

\begin{equation}
t'\mathbb{V}\left(b_{0}\left(W\right)+\Delta_{*}^{\left(J\right)}\left(k^{\left(J\right)}\left(W\right)-\mu^{\left(J\right)}\right)-\beta_{0}\right)t\leq C_{1}\zeta^{2},\label{eq: term_1_bound}
\end{equation}
with $C_{1}$ a constant.

Next consider the second term in (\ref{eq: var_dif_2}). Applying
the Cauchy-Schwarz inequality to this term yields
\begin{align*}
\left|t'\mathbb{E}\left[\left(b_{0}\left(W\right)-\beta_{0}\right)\left\{ b_{0}\left(W\right)+\Delta_{*}^{\left(J\right)}\left(k^{\left(J\right)}\left(W\right)-\mu^{\left(J\right)}\right)-\beta_{0}\right\} '\right]t\right|\leq & \mathbb{V}\left(t'b_{0}\left(W\right)\right)^{1/2}\\
 & \times\mathbb{V}\left(t'\left\{ b_{0}\left(W\right)\right.\right.\\
 & \left.\left.+\Delta_{*}^{\left(J\right)}\left(k^{\left(J\right)}\left(W\right)-\mu^{\left(J\right)}\right)-\beta_{0}\right\} \right)^{1/2}
\end{align*}
Again invoking completeness of the sequence $\left\{ k_{j}\left(W\right)\right\} _{j=1}^{\infty}$,
and also boundedness of the variance of $b_{0}\left(W\right)$, we
then get
\begin{equation}
\left|t'\mathbb{E}\left[\left(b_{0}\left(W\right)-\beta_{0}\right)\left\{ b_{0}\left(W\right)+\Delta_{*}^{\left(J\right)}\left(k^{\left(J\right)}\left(W\right)-\mu^{\left(J\right)}\right)-\beta_{0}\right\} '\right]t\right|\leq C_{2}\zeta,\label{eq: term_2_bound}
\end{equation}
with $C_{2}$ a constant (which depends on $\mathbb{V}\left(b_{0}\left(W\right)\right)$).

Finally consider the third term in (\ref{eq: var_dif_1}). To analyze
this term we start by writing the linear predictor approximation error
of $\left(R^{\left(J\right)}\right)'\lambda_{*}^{\left(J\right)}$
for $a_{0}\left(W\right)+X'\left(b_{0}\left(W\right)-\beta_{0}\right)$
as
\begin{align*}
\epsilon^{\left(J\right)}= & \left\{ a_{0}\left(W\right)+X'\left(b_{0}\left(W\right)-\beta_{0}\right)-\left(R^{\left(J\right)}\right)'\lambda_{*}^{\left(J\right)}\right\} \\
= & a\left(W\right)-\alpha_{*}^{\left(J\right)}-\left(k^{\left(J\right)}\left(W\right)-\mu^{\left(J\right)}\right)'\gamma_{*}^{\left(J\right)}\\
 & +X'\left(b_{0}\left(W\right)-\beta_{0}-\Delta_{*}^{\left(J\right)}\left(k^{\left(J\right)}\left(W\right)-\mu^{\left(J\right)}\right)\right)\\
= & \left(1,X'\right)\delta^{\left(J\right)}\left(W\right),
\end{align*}
with the final equality following from definition (\ref{eq: delta(W)_J}).
The EI and the fact that, for $\mathbf{a}$ and $\mathbf{b}$ $m\times1$
vectors $\left\Vert \mathbf{a}\mathbf{b}'\right\Vert _{F}=\left\Vert \mathbf{a}\right\Vert \left\Vert \mathbf{b}\right\Vert $,
then gives
\begin{align*}
\left\Vert \mathbb{E}\left[\left(\tilde{\epsilon}^{\left(J\right)}-\Pi_{\tilde{\epsilon}\mathbb{S}}^{\left(J\right)}\mathbb{S}\right)\left(\tilde{\epsilon}^{\left(J\right)}-\Pi_{\tilde{\epsilon}\mathbb{S}}^{\left(J\right)}\mathbb{S}\right)'\right]\right\Vert  & \leq\mathbb{E}\left[\left\Vert \tilde{\epsilon}^{\left(J\right)}-\Pi_{\tilde{\epsilon}\mathbb{S}}^{\left(J\right)}\mathbb{S}\right\Vert ^{2}\right]
\end{align*}
with
\[
\tilde{\epsilon}^{\left(J\right)}=v_{0}\left(W\right)^{-1}\left(X-e_{0}\left(W\right)\right)\left\{ \left(1,X'\right)\delta^{\left(J\right)}\left(W\right)\right\} .
\]
By the norm-reducing property of projection and Schwarz Matrix Inequality
(SMI) we further get that 
\begin{align*}
\left\Vert \tilde{\epsilon}^{\left(J\right)}-\Pi_{\tilde{\epsilon}\mathbb{S}}^{\left(J\right)}\mathbb{S}\right\Vert  & \leq\left\Vert \tilde{\epsilon}^{\left(J\right)}\right\Vert \\
 & =\left\Vert v_{0}\left(W\right)^{-1}\left(X-e_{0}\left(W\right)\right)\left\{ \left(1,X'\right)\delta^{\left(J\right)}\left(W\right)\right\} \right\Vert \\
 & \leq\left\Vert v_{0}\left(W\right)^{-1}\left(X-e_{0}\left(W\right)\right)\left\{ \left(1,X'\right)\right\} \right\Vert \left\Vert \delta^{\left(J\right)}\left(W\right)\right\Vert .
\end{align*}
Applying the expectation operator, invoking Assumption \ref{ass: overlap},
and using the compact support assumption for $\left(W,X\right)$,
finally gives
\begin{align}
\mathbb{E}\left[\left\Vert \tilde{\epsilon}^{\left(J\right)}-\Pi_{\tilde{\epsilon}\mathbb{S}}^{\left(J\right)}\mathbb{S}\right\Vert ^{2}\right]\leq & \mathbb{E}\left[\left\Vert v_{0}\left(W\right)^{-1}\left(X-e_{0}\left(W\right)\right)\left\{ \left(1,X'\right)\right\} \right\Vert ^{2}\left\Vert \delta^{\left(J\right)}\left(W\right)\right\Vert ^{2}\right]\nonumber \\
\leq & C_{3}\mathbb{E}\left[\left\Vert \delta^{\left(J\right)}\left(W\right)\right\Vert ^{2}\right]\nonumber \\
\leq & C_{3}\zeta^{2}\label{eq: term_3_bound}
\end{align}
with $C_{3}=\underset{w,x\in\mathbb{W},\mathbb{X}}{\sup}\left\Vert v_{0}\left(W\right)^{-1}\left(X-e_{0}\left(W\right)\right)\left\{ \left(1,X'\right)\right\} \right\Vert ^{2}$.

Applying the TI to (\ref{eq: var_dif_2}) and using terms (\ref{eq: term_1_bound}),
(\ref{eq: term_2_bound}) and (\ref{eq: term_3_bound}) then gives
the bound
\begin{equation}
0\leq t'\mathcal{I}^{\left(J\right)}\left(\beta_{0}\right)^{-1}t-t'\mathcal{I}\left(\beta_{0}\right)^{-1}t\leq\left(C_{1}+C_{3}\right)\zeta^{2}+C_{2}\zeta.\label{eq: inefficient_ineq_3}
\end{equation}
Since $\zeta$ is arbitrary the limit of the difference in (\ref{eq: inefficient_ineq_3})
is zero.

\bibliographystyle{apalike2}
\bibliography{0_Users_bgraham_Dropbox_Research_Networks_Networks_Book_Reference_BibTex_Networks_References}

\begin{thebibliography}{}

\bibitem[Angrist, 1998]{Angrist_EM98}
Angrist, J.~D. (1998).
\newblock Estimating the labor market impact of voluntary military service
  using social security data on military applicants.
\newblock {\em Econometrica}, 66(2), 249 -- 288.

\bibitem[Angrist et~al., 2000]{Angrist_Graddy_Imbens_ReStud00}
Angrist, J.~D., Graddy, K., \& Imbens, G.~W. (2000).
\newblock The interpretation of instrumental variables estimators in
  simultaneous equations models with an application to the demand for fish.
\newblock {\em Review of Economic Studies}, 67(3), 499 -- 527.

\bibitem[Angrist \& Krueger, 1999]{Angrist_Krueger_HLE99}
Angrist, J.~D. \& Krueger, A.~B. (1999).
\newblock {\em Handbook of Labor Economics}, volume~3, chapter Empirical
  strategies in labor economics, (pp.\ 1277 -- 1366.).
\newblock North-Holland: Amsterdam.

\bibitem[Angrist \& Pischke, 2009]{Angrist_Pischke_MHE09}
Angrist, J.~D. \& Pischke, J.-S. (2009).
\newblock {\em Mostly Harmless Econometrics}.
\newblock Princeton, NJ: Princeton University Press.

\bibitem[Bang \& Robins, 2005]{Bang_Robins_BM05}
Bang, H. \& Robins, J.~M. (2005).
\newblock Doubly robust estimation in missing data and causal inference models.
\newblock {\em Biometrics}, 61(4), 962 -- 973.

\bibitem[Belloni et~al., 2014]{Belloni_et_al_ReStud14}
Belloni, A., Chernozhukov, V., \& Hansen, C. (2014).
\newblock Inference on treatment effects after selection among high-dimensional
  controls.
\newblock {\em Review of Economic Studies}, 81(2), 608 -- 650.

\bibitem[Bickel et~al., 1993]{Bickel_et_al_Bk93}
Bickel, P.~J., Klaassen, C. A.~J., Ritov, Y., \& Wellner, J.~A. (1993).
\newblock {\em Efficient and Adaptive Estimation for Semiparametric Models}.
\newblock New York: Springer-Verlag.

\bibitem[Blundell \& Powell, 2003]{Blundel_Powell_WC03}
Blundell, R. \& Powell, J.~L. (2003).
\newblock {\em Advances in Economics and Econometrics: Theory and Applications,
  Eighth World Congress}, volume~2, chapter Endogeneity in nonparametric and
  semiparametric regression models, (pp.\ 312 -- 357).
\newblock Cambridge University Press.

\bibitem[Cattaneo, 2010]{Cattaneo_JOE10}
Cattaneo, M.~D. (2010).
\newblock Efficient semiparametric estimation of multi-valued treatment effects
  under ignorability.
\newblock {\em Journal of Econometrics}, 155(2), 138 -- 154.

\bibitem[Chamberlain, 1984]{Chamberlain_HBE84}
Chamberlain, G. (1984).
\newblock {\em Handbook of Econometrics}, volume~2, chapter Panel Data, (pp.\
  1247 -- 1318).
\newblock North-Holland: Amsterdam.

\bibitem[Chamberlain, 1986]{Chamberlain_WP86}
Chamberlain, G. (1986).
\newblock {\em Notes on Semiparametric Regression}.
\newblock Working paper, University of Wisconsin - Madison.

\bibitem[Chamberlain, 1987]{Chamberlain_JE87}
Chamberlain, G. (1987).
\newblock Asymptotic efficiency in estimation with conditional moment
  restrictions.
\newblock {\em Journal of Econometrics}, 34(3), 305 -- 334.

\bibitem[Chamberlain, 1992]{Chamberlain_EM92}
Chamberlain, G. (1992).
\newblock Efficiency bounds for semiparametric regression.
\newblock {\em Econometrica}, 60(3), 567 -- 596.

\bibitem[Chen et~al., 2008]{Chen_Hong_Tarozz_AS08}
Chen, X., Hong, H., \& Tarozzi, A. (2008).
\newblock Semiparametric efficiency in gmm models with auxiliary data.
\newblock {\em Annals of Statistics}, 36(2), 808 -- 843.

\bibitem[Fr{\"o}lich, 2004]{Frolich_ER04}
Fr{\"o}lich, M. (2004).
\newblock A note on the role of the propensity score for estimating average
  treatment effects.
\newblock {\em Econometric Reviews}, 23(2), 167 -- 174.

\bibitem[Goldberger, 1991]{Goldberger_ACE91}
Goldberger, A.~S. (1991).
\newblock {\em A Course in Econometrics}.
\newblock Cambridge, MA: Harvard University Press.

\bibitem[Gottfried \& Kirksey, 2017]{Gottfried_Kirksey_ER17}
Gottfried, M.~A. \& Kirksey, J.~J. (2017).
\newblock ``when'' students miss school: the role of timing of absenteeism on
  students' test performance.
\newblock {\em Educational Researcher}, 46(3), 119 -- 130.

\bibitem[Graham, 2011]{Graham_EM11}
Graham, B.~S. (2011).
\newblock Efficiency bounds for missing data models with semiparametric
  restrictions.
\newblock {\em Econometrica}, 79(2), 437 -- 452.

\bibitem[Graham et~al., 2010]{Graham_Imbens_Ridder_NBER10}
Graham, B.~S., Imbens, G.~W., \& Ridder, G. (2010).
\newblock {\em Measuring the effects of segregation in the presence of social
  spillovers: a nonparametric approach}.
\newblock Working Paper 16499, NBER.

\bibitem[Graham et~al., 2012]{Graham_Pinto_Egel_ReStud12}
Graham, B.~S., Pinto, C., \& Egel, D. (2012).
\newblock Inverse probability tilting for moment condition models with missing
  data.
\newblock {\em Review of Economic Studies}, 79(3), 1053 -- 1079.

\bibitem[Graham et~al., 2016]{Graham_Pinto_Egel_JBES16}
Graham, B.~S., Pinto, C., \& Egel, D. (2016).
\newblock Efficient estimation of data combination models by the method of
  auxiliary-to-study tilting (ast).
\newblock {\em Journal of Business and Economic Statistics}, 31(2), 288 -- 301.

\bibitem[Groves \& Rothenberg, 1969]{Groves_Rothenberg_BM69}
Groves, T. \& Rothenberg, T. (1969).
\newblock A note on the expected value of an inverse matrix.
\newblock {\em Biometrika}, 56(3), 690 -- 691.

\bibitem[Hahn, 1998]{Hahn_EM98}
Hahn, J. (1998).
\newblock On the role of the propensity score in efficient semiparametric
  estimation of average treatment effects.
\newblock {\em Econometrica}, 66(2), 315 -- 331.

\bibitem[Hansen, 2018]{Hansen_Book18}
Hansen, B. (2018).
\newblock Econometrics.

\bibitem[Hastie \& Tibshirani, 1993]{Hastie_Tibshirani_JRSS93}
Hastie, T. \& Tibshirani, R. (1993).
\newblock Varying-coefficient models.
\newblock {\em Journal of the Royal Statistical Society B}, 55(4), 757 -- 796.

\bibitem[Henderson \& Searle, 1981]{Henderson_Searle_SIAM81}
Henderson, H.~V. \& Searle, S.~R. (1981).
\newblock On deriving the inverse of a sum of matrices.
\newblock {\em SIAM Review}, 23(1), 53 -- 60.

\bibitem[Hirano \& Imbens, 2001]{Hirano_Imbens_HSORM01}
Hirano, K. \& Imbens, G.~W. (2001).
\newblock Estimation of causal effects using propensity score weighting: an
  application to data on right heart catheterization.
\newblock {\em Health Services and Outcomes Research Methodology}, 2(3-4), 259
  -- 278.

\bibitem[Hirano \& Imbens, 2004]{Hirano_Imbens_BM04}
Hirano, K. \& Imbens, G.~W. (2004).
\newblock {\em Applied Bayesian Modelling and Causal Inference from Missing
  Data Perspectives}, chapter The propensity score with continuous treatments,
  (pp.\ 73 -- 84).
\newblock John Wiley \& Sons, Inc.: New York.

\bibitem[Hirano et~al., 2003]{Hirano_et_al_EM03}
Hirano, K., Imbens, G.~W., \& Ridder, G. (2003).
\newblock Efficient estimation of average treatment effects using the estimated
  propensity score.
\newblock {\em Econometrica}, 71(4), 1161 -- 1189.

\bibitem[Hitomi et~al., 2008]{Hitomi_et_al_ET08}
Hitomi, K., Nishiyama, Y., \& Okui, R. (2008).
\newblock A puzzling phenomenon in semiparametric estimation problems with
  infinite-dimensional nuisance parameters.
\newblock {\em Econometric Theory}, 24(6), 1717 -- 1728.

\bibitem[Imbens, 2000]{Imbens_BM00}
Imbens, G.~W. (2000).
\newblock The role of the propensity score in estimating dose-response functio.
\newblock {\em Biometrika}, 87(3), 706 -- 710.

\bibitem[Imbens \& Rubin, 2015]{Imbens_Rubin_CIBook15}
Imbens, G.~W. \& Rubin, D.~B. (2015).
\newblock {\em Causal Inference for Statistics, Social, and Biomedical
  Sciences: An Introduction}.
\newblock Cambridge: Cambridge University Press.

\bibitem[Kline, 2014]{Kline_EL14}
Kline, P. (2014).
\newblock A note on variance estimation for the oaxaca estimator of average
  treatment effects.
\newblock {\em Economics Letters}, 122(3), 428 -- 431.

\bibitem[Newey, 1990]{Newey_JAE90}
Newey, W.~K. (1990).
\newblock Semiparametric efficiency bounds.
\newblock {\em Journal of Applied Econometrics}, 5(2), 99 -- 135.

\bibitem[Newey, 1994]{Newey_ET94b}
Newey, W.~K. (1994).
\newblock Kernel estimation of partial means and a general variance estimator.
\newblock {\em Econometric Theory}, 10(2), 233 -- 253.

\bibitem[Newey \& McFadden, 1994]{Newey_McFadden_HBE94}
Newey, W.~K. \& McFadden, D. (1994).
\newblock {\em Handbook of Econometrics}, volume~4, chapter Large sample
  estimation and hypothesis testing, (pp.\ 2111 -- 2245).
\newblock North-Holland: Amsterdam.

\bibitem[Pencavel, 1986]{Pencavel_HLE86}
Pencavel, J. (1986).
\newblock {\em Handbook of Labor Economics}, volume~1, chapter Labor supply of
  men: a survey, (pp.\ 3 -- 102).
\newblock North-Holland: Amsterdam.

\bibitem[Robins et~al., 1992]{Robins_Mark_Newey_BM92}
Robins, J.~M., Mark, S.~D., \& Newey, W.~K. (1992).
\newblock Estimating exposure effects by modelling the expectation of exposure
  conditional on confounders.
\newblock {\em Biometrics}, 48(2), 479 -- 495.

\bibitem[Robins et~al., 1994]{Robins_Rotnitzky_Zhao_JASA94}
Robins, J.~M., Rotnitzky, A., \& Zhao, L.~P. (1994).
\newblock Estimation of regression coefficients when some regressors are not
  always observed.
\newblock {\em Journal of American Statistical Association}, 89(427), 846 --
  866.

\bibitem[Robinson, 1988]{Robinson_EM88}
Robinson, P.~M. (1988).
\newblock Root-n-consistent semiparametric regression.
\newblock {\em Econometrica}, 56(4), 931 -- 954.

\bibitem[Rosenbaum \& Rubin, 1983]{Rosenbaum_Rubin_BM83}
Rosenbaum, P.~R. \& Rubin, D.~B. (1983).
\newblock The central role of the propensity score in observational studies for
  causal effects.
\newblock {\em Biometrika}, 70(1), 41 -- 55.

\bibitem[Ruud, 1986]{Ruud_JOE86}
Ruud, P.~A. (1986).
\newblock Consistent estimation of limited dependent variable models despite
  misspecification of distribution.
\newblock {\em Journal of Econometrics}, 32(1), 157--187.

\bibitem[Scharfstein et~al., 1999]{Scharfstein_Rotnitzky_Robins_JASA99b}
Scharfstein, D.~O., Rotnitzky, A., \& Robins, J.~M. (1999).
\newblock Adjusting for nonignorable drop-out using semiparametric nonresponse
  models: rejoinder.
\newblock {\em Journal of American Statistical Association}, 94(448), 1135 --
  1146.

\bibitem[Sloczynski, 2015]{Sloczynski_OBES15}
Sloczynski, T. (2015).
\newblock The oaxaca--blinder unexplained component as a treatment effects
  estimator.
\newblock {\em Oxford Bulletin of Economics and Statistics}, 77(4), 588 -- 604.

\bibitem[Sloczynski, 2017]{Sloczynski_WP18}
Sloczynski, T. (2017).
\newblock {\em A general weighted average representation of the ordinary and
  two-stage least squares estimands}.
\newblock Working paper, Brandies University.

\bibitem[Tsiatis, 2006]{Tsiatis_Book06}
Tsiatis, A.~A. (2006).
\newblock {\em Semiparametric Theory and Missing Data}.
\newblock New York: Springer.

\bibitem[Wooldridge, 1999]{Wooldridge_JE99}
Wooldridge, J.~M. (1999).
\newblock Distribution-free estimation of some nonlinear panel data models.
\newblock {\em Journal of Econometrics}, 90(1), 77 -- 97.

\bibitem[Wooldridge, 2004]{Wooldridge_CWP04}
Wooldridge, J.~M. (2004).
\newblock {\em Estimating average partial effects under conditional moment
  independence assumptions}.
\newblock Working Paper CWP03/04, CeMMAP.

\bibitem[Wooldridge, 2005]{Wooldridge_IIEM05}
Wooldridge, J.~M. (2005).
\newblock {\em Identification and inference for econometric models}, chapter
  Unobserved heterogeneity and the estimation of average partial effects, (pp.\
  27 -- 55).
\newblock Number~3. Cambridge University Press: Cambridge.

\bibitem[Wooldridge, 2007]{Wooldridge_JE07}
Wooldridge, J.~M. (2007).
\newblock Inverse probability weighted estimation for general missing data
  problems.
\newblock {\em Journal of Econometrics}, 141(2), 1281 -- 1301.

\bibitem[Wooldridge, 2010]{Wooldridge_EACSPDBook10}
Wooldridge, J.~M. (2010).
\newblock {\em Econometric Analysis of Cross Section and Panel Data}.
\newblock Cambridge, MA: MIT Press, 2nd edition.

\bibitem[Yitzhaki, 1996]{Yitzhaki_JBES96}
Yitzhaki, S. (1996).
\newblock On using linear regressions in welfare economics.
\newblock {\em Journal of Business and Economic Statistics}, 14(4), 478 -- 486.

\bibitem[Yule, 1899]{Yule_JRSS1899}
Yule, G.~U. (1899).
\newblock An investigation into the causes of changes in pauperism in england,
  chiefly during the last two intercensal decades (part i.).
\newblock {\em Journal of the Royal Statistical Society}, 62(6), 249 -- 295.

\end{thebibliography}

\pagebreak{}

\section{\label{sec: Supplemental_Web_Appendix}Supplemental web appendix
for ``Semiparametrically efficient estimation of the average linear
regression function'' by Bryan Graham and Cristine Pinto}

This supplemental web appendix contains proofs of the results not
included in the main appendix as well as additional detailed calculations
for some proof steps. All notation is as defined in the main text
and/or appendix unless explicitly noted otherwise. Equation numbering
continues in sequence with that established in the main text and its
appendix. 

\subsection*{Proof of Proposition \ref{prop: wgt_der}}

Begin by noting that under Assumption \ref{ass: wgt_der} we have
\begin{align*}
h\left(x,U\right) & =\underline{h}\left(U\right)+\int_{\underline{x}}^{x}\frac{\partial h\left(t,U\right)}{\partial x}\mathrm{d}t\\
 & =\underline{h}\left(U\right)+\int_{\underline{x}}^{\bar{x}}\frac{\partial h\left(t,U\right)}{\partial x}\mathbf{1}\left(x\geq t\right)\mathrm{d}t,
\end{align*}
which, invoking conditional independence yields
\begin{align*}
\mathbb{E}\left[\frac{X-e_{0}\left(W\right)}{v_{0}\left(W\right)}\underline{h}\left(U\right)\right] & =\mathbb{E}\left[\frac{X-e_{0}\left(W\right)}{v_{0}\left(W\right)}\mathbb{E}\left[\left.\underline{h}\left(U\right)\right|W,X\right]\right]\\
 & =\mathbb{E}\left[\frac{X-e_{0}\left(W\right)}{v_{0}\left(W\right)}\mathbb{E}\left[\left.\underline{h}\left(U\right)\right|W\right]\right]\\
 & =\mathbb{E}\left[\mathbb{E}\left[\left.\frac{X-e_{0}\left(W\right)}{v_{0}\left(W\right)}\right|W\right]\mathbb{E}\left[\left.\underline{h}\left(U\right)\right|W\right]\right]\\
 & =\mathbb{E}\left[v_{0}\left(W\right)^{-1}\mathbb{E}\left[\left.X-e_{0}\left(W\right)\right|W\right]\mathbb{E}\left[\left.\underline{h}\left(U\right)\right|W\right]\right]\\
 & =\mathbb{E}\left[v_{0}\left(W\right)^{-1}\cdot0\cdot\mathbb{E}\left[\left.\underline{h}\left(U\right)\right|W\right]\right]\\
 & =0.
\end{align*}

Using this result we can re-write the $\beta_{0}$ estimand as follows:

\begin{align*}
\mathbb{E}\left[\frac{X-e_{0}\left(W\right)}{v_{0}\left(W\right)}Y\right] & =\mathbb{E}\left[\frac{X-e_{0}\left(W\right)}{v_{0}\left(W\right)}\int_{\underline{x}}^{\bar{x}}\frac{\partial h\left(t,U\right)}{\partial x}\mathbf{1}\left(X\geq t\right)\mathrm{d}t\right]\\
 & =\mathbb{E}\left[\int_{\underline{x}}^{\bar{x}}\frac{\partial h\left(t,U\right)}{\partial x}\mathbf{1}\left(X\geq t\right)\frac{X-e_{0}\left(W\right)}{v_{0}\left(W\right)}\mathrm{d}t\right]\\
 & =\mathbb{E}\left[\int_{\underline{x}}^{\bar{x}}\mathbb{E}\left[\left.\frac{\partial h\left(t,U\right)}{\partial x}\right|W,X\right]\mathbf{1}\left(X\geq t\right)\frac{X-e_{0}\left(W\right)}{v_{0}\left(W\right)}\mathrm{d}t\right]\\
 & =\mathbb{E}\left[\int_{\underline{x}}^{\bar{x}}\mathbb{E}\left[\left.\frac{\partial h\left(t,U\right)}{\partial x}\right|W\right]\mathbf{1}\left(X\geq t\right)\frac{X-e_{0}\left(W\right)}{v_{0}\left(W\right)}\mathrm{d}t\right]\\
 & =\mathbb{E}\left[\int_{\underline{x}}^{\bar{x}}\mathbb{E}\left[\left.\frac{\partial h\left(t,U\right)}{\partial x}\right|W\right]\mathbb{E}\left[\left.\frac{X-e_{0}\left(W\right)}{v_{0}\left(W\right)}\right|W,X\geq t\right]\left(1-F_{\left.X\right|W}\left(\left.t\right|W\right)\right)\mathrm{d}t\right].
\end{align*}
Next observe that
\begin{align*}
v_{0}\left(w\right) & =\mathbb{E}\left[\left.X\left(X-e_{0}\left(W\right)\right)\right|W=w\right]\\
 & =\mathbb{E}\left[\left.\int_{\underline{x}}^{\bar{x}}\left(X-e_{0}\left(W\right)\right)\mathrm{d}t\right|W=w\right]\\
 & =\int_{\underline{x}}^{\bar{x}}\mathbb{E}\left[\left.X-e_{0}\left(W\right)\right|W=w,X\geq t\right]\left(1-F_{\left.X\right|W}\left(\left.t\right|w\right)\right)\mathrm{d}t.
\end{align*}
Putting all these pieces together we have
\begin{align*}
\mathbb{E}\left[\frac{X-e_{0}\left(W\right)}{v_{0}\left(W\right)}Y\right]= & \mathbb{E}\left[\int_{\underline{x}}^{\bar{x}}\mathbb{E}\left[\left.\frac{\partial h\left(t,U\right)}{\partial x}\right|W\right]\frac{\mathbb{E}\left[\left.X-e_{0}\left(W\right)\right|W,X\geq t\right]\left(1-F_{\left.X\right|W}\left(\left.t\right|W\right)\right)}{\int_{\underline{x}}^{\bar{x}}\mathbb{E}\left[\left.X-e_{0}\left(W\right)\right|W,X\geq v\right]\left(1-F_{\left.X\right|W}\left(\left.v\right|W\right)\right)\mathrm{d}v}\mathrm{d}t\right]\\
= & \mathbb{E}\left[\int_{\underline{x}}^{\bar{x}}\int_{-\infty}^{\infty}\left(\frac{\partial h\left(t,u\right)}{\partial x}f_{\left.U\right|W,X}\left(\left.u\right|w,t\right)\mathrm{d}u\right.\right.\\
 & \left.\left.\times\frac{1}{f_{\left.X\right|W}\left(\left.t\right|W\right)}\frac{\mathbb{E}\left[\left.X-e_{0}\left(W\right)\right|W,X\geq t\right]\left(1-F_{\left.X\right|W}\left(\left.t\right|W\right)\right)}{\int_{\underline{x}}^{\bar{x}}\mathbb{E}\left[\left.X-e_{0}\left(W\right)\right|W,X\geq v\right]\left(1-F_{\left.X\right|W}\left(\left.v\right|W\right)\right)\mathrm{d}v}f_{\left.X\right|W}\left(\left.t\right|W\right)\right)\mathrm{d}t\right]\\
= & \mathbb{E}\left[\int_{\underline{x}}^{\bar{x}}\int_{-\infty}^{\infty}\left(\frac{\partial h\left(t,u\right)}{\partial x}\frac{1}{f_{\left.X\right|W}\left(\left.t\right|W\right)}\right.\right.\\
 & \left.\left.\times\frac{\mathbb{E}\left[\left.X-e_{0}\left(W\right)\right|W,X\geq t\right]\left(1-F_{\left.X\right|W}\left(\left.t\right|W\right)\right)}{\int_{\underline{x}}^{\bar{x}}\mathbb{E}\left[\left.X-e_{0}\left(W\right)\right|W,X\geq v\right]\left(1-F_{\left.X\right|W}\left(\left.v\right|W\right)\right)\mathrm{d}v}f_{\left.U,X\right|W}\left(\left.u,t\right|W\right)\right)\mathrm{d}u\mathrm{d}t\right]\\
 & =\mathbb{E}\left[\omega\left(W,X\right)\frac{\partial h\left(X,U\right)}{\partial x}\right],
\end{align*}
with
\[
\omega\left(w,x\right)=\frac{1}{f_{\left.X\right|W}\left(\left.x\right|w\right)}\frac{\mathbb{E}\left[\left.X-e_{0}\left(W\right)\right|W=w,X\geq x\right]\left(1-F_{\left.X\right|W}\left(\left.x\right|w\right)\right)}{\int_{\underline{x}}^{\bar{x}}\mathbb{E}\left[\left.X-e_{0}\left(W\right)\right|W=w,X\geq v\right]\left(1-F_{\left.X\right|W}\left(\left.v\right|w\right)\right)\mathrm{d}v}.
\]

\subsection*{Proof of Corollary \ref{cor: redundancy_of_f(x|w)}}

Let $f\left(\left.x\right|w;\phi\right)$ be a known parametric family
of conditional distributions for $X$ given $W$. Let $f_{0}\left(\left.x\right|w\right)=f\left(\left.x\right|w;\phi\right)$
at some unique $\phi=\phi_{0}$. Relative to that considered in Theorem
\ref{thm: SEB}, the parametric submodel changes to
\[
f\left(w,x,y;\eta\right)=f\left(\left.y\right|w,x\right)f\left(\left.x\right|w;\phi\left(\eta\right)\right)f\left(w;\eta\right)
\]
with an associated score vector of
\begin{equation}
s_{\eta}\left(w,x,y;\eta\right)=s_{\eta}\left(\left.y\right|w,x;\eta\right)+\left(\frac{\partial\phi\left(\eta\right)}{\partial\eta'}\right)'\mathbb{S}_{\phi}\left(\left.x\right|w;\phi\right)+t_{\eta}\left(w;\eta\right),\label{eq: score_parametric_f(x|w)}
\end{equation}
where $\mathbb{S}_{\phi}\left(\left.x\right|w;\phi\right)$ is the
score function associated with the parametric conditional log-likelihood
for $\phi$.

From (\ref{eq: score_parametric_f(x|w)}), and the usual (conditional)
mean zero properties of score functions, the tangent set is evidently
\[
\mathcal{T}=\left\{ s\left(\left.y\right|w,x\right)+\mathbf{c}\mathbb{S}_{\phi}\left(\left.x\right|w\right)+t\left(w\right)\right\} 
\]
where $\mathbb{S}_{\phi}\left(\left.x\right|w\right)=\mathbb{S}_{\phi}\left(\left.x\right|w;\phi_{0}\right)$,
$\mathbf{c}$ is a matrix of constants, and
\[
\mathbb{E}\left[\left.s\left(\left.Y\right|W,X\right)\right|W,X\right]=\mathbb{E}\left[\left.\mathbb{S}_{\phi}\left(\left.X\right|W\right)\right|W\right]=\mathbb{E}\left[t\left(W\right)\right]=0.
\]

To show pathwise differentiability, begin by noting that $\beta\left(\eta\right)$
continues to equal (\ref{eq: beta(eta)}), however $b\left(w;\eta\right)$
now satisfies the modified conditional moment restriction 
\begin{equation}
\int\int\left(\begin{array}{c}
1\\
x
\end{array}\right)\left(y-a\left(w;\eta\right)-x'b\left(w;\eta\right)\right)f\left(\left.y\right|w,x;\eta\right)f\left(\left.x\right|w;\phi\left(\eta\right)\right)\mathrm{d}x\mathrm{d}y=0.\label{eq: CLP(eta)_parametric_f(x|w)}
\end{equation}

We can derive a close-form expression for $\frac{\partial b\left(w;\eta_{0}\right)}{\partial\eta'}$
in (\ref{eq: pathwise_1}) by differentiating (\ref{eq: CLP(eta)_parametric_f(x|w)})
with respect to $\eta$ (and evaluating at $\eta=\eta_{0}$): 
\begin{eqnarray*}
-\int\int\left(\begin{array}{c}
1\\
x
\end{array}\right)\frac{\partial a\left(w;\eta_{0}\right)}{\partial\eta'}f\left(\left.y\right|w,x;\eta_{0}\right)f\left(\left.x\right|w;\phi_{0}\right)\mathrm{d}x\mathrm{d}y\\
-\int\int\left(\begin{array}{c}
x'\\
xx'
\end{array}\right)\frac{\partial b\left(w;\eta_{0}\right)}{\partial\eta'}f\left(\left.y\right|w,x;\eta_{0}\right)f\left(\left.x\right|w;\phi_{0}\right)\mathrm{d}x\mathrm{d}y\\
+\left(\int\int\left(\begin{array}{c}
1\\
x
\end{array}\right)\left(y-x'b\left(w;\eta_{0}\right)\right)\left\{ s_{\eta}\left(\left.y\right|w,x;\eta_{0}\right)+\left(\frac{\partial\phi\left(\eta_{0}\right)}{\partial\eta'}\right)'\mathbb{S}_{\phi}\left(\left.x\right|w\right)\right\} \right.\\
\left.\times f\left(\left.y\right|w,x;\eta_{0}\right)f\left(\left.x\right|w;\phi_{0}\right)\mathrm{d}x\mathrm{d}y\right) & = & 0
\end{eqnarray*}
Analogous to the corresponding calculations given in the proof of
Theorem \ref{thm: SEB} we can solve to get
\begin{eqnarray*}
\left(\begin{array}{c}
\frac{\partial a\left(w;\eta_{0}\right)}{\partial\eta'}\\
\frac{\partial b\left(w;\eta_{0}\right)}{\partial\eta'}
\end{array}\right) & = & \left(\begin{array}{cc}
1 & -e\left(w;\phi_{0}\right)'v\left(w;\phi_{0}\right)^{-1}\\
-v\left(w;\phi_{0}\right)^{-1}e\left(w;\phi_{0}\right) & v\left(w;\phi_{0}\right)^{-1}
\end{array}\right)\\
 &  & \times\mathbb{E}\left[\left(\begin{array}{c}
Y-a\left(W;\eta_{0}\right)-X'b\left(W;\eta_{0}\right)\\
X\left(Y-a\left(W;\eta_{0}\right)-X'b\left(W;\eta_{0}\right)\right)
\end{array}\right)\right.\\
 &  & \left.\left.\times\left\{ s_{\eta}\left(\left.Y\right|W,X;\eta_{0}\right)+\left(\frac{\partial\phi\left(\eta_{0}\right)}{\partial\eta'}\right)'\mathbb{S}_{\phi}\left(\left.X\right|W\right)\right\} \right|W=w\right].
\end{eqnarray*}
Plugging the second row of the above expression into (\ref{eq: pathwise_1}),
which remains unchanged relative to its form in the proof of Theorem
\ref{thm: SEB}, we get

\begin{eqnarray}
\frac{\partial\beta\left(\eta_{0}\right)}{\partial\eta'} & = & \mathbb{E}\left[v\left(W;\phi_{0}\right)^{-1}\left(X-e\left(W;\phi_{0}\right)\right)\left(Y-a_{0}\left(W\right)-X'b_{0}\left(W\right)\right)\right.\nonumber \\
 &  & \left.\times\left\{ s_{\eta}\left(\left.Y\right|W,X;\eta_{0}\right)+\left(\frac{\partial\phi\left(\eta_{0}\right)}{\partial\eta'}\right)'\mathbb{S}_{\phi}\left(\left.X\right|W\right)\right\} \right]\nonumber \\
 &  & +\mathbb{E}\left[b_{0}\left(W\right)t_{\eta}\left(W\right)\right].\label{eq: pathwise_parametric_f(x|w)}
\end{eqnarray}
Now observe that (\ref{eq: EfficientInfluenceFunction_CRC}) remains
a pathwise derivative. Furthermore (\ref{eq: EfficientInfluenceFunction_CRC})
continues to lie in the tangent space with its first component playing
the role of $s\left(\left.x,y\right|w\right)=s\left(\left.y\right|w,x\right)+\mathbf{c}\mathbb{S}_{\phi}\left(\left.x\right|w\right)$
and its second component that of $t\left(w\right)$. The claim again
follows from Theorem 3.1 of \citet{Newey_JAE90}.

\subsection*{Detailed calculations for proof of Theorem (\ref{thm: Large-Sample})}

Let $m\left(Z_{i},\theta\right)$ be the $\left(L+J+1+J+JK+K\right)\times1$
vector of moment conditions as defined in the main text. In this appendix
we work with the more refined partition of this vector:
\begin{align}
m_{1}(X_{i},W_{i},\phi)= & \mathbb{S}_{\phi}\left(\left.X_{i}\right|W_{i};\phi\right)\label{eq: m1}\\
m_{2}(W_{i},\mu_{W})= & W_{i}-\mu_{W}\label{eq: m2}\\
m_{3}(Z_{i},\mu_{W},\lambda,\beta)= & U_{i}\left(\mu_{W},\lambda,\beta\right)\label{eq: m3}\\
m_{4}(Z_{i},\mu_{W},\lambda,\beta)= & \left(W_{i}-\mu_{W}\right)U_{i}\left(\mu_{W},\lambda,\beta\right)\label{eq: m4}\\
m_{5}(Z_{i},\mu_{W},\lambda,\beta)= & \left(\left(W_{i}-\mu_{W}\right)\varotimes X_{i}\right)U_{i}\left(\mu_{W},\lambda,\beta\right)\label{eq: m5}\\
m_{5}(Z_{i},\phi,\mu_{W},\lambda,\beta)= & v\left(W;\phi\right)^{-1}\left(X-e\left(W;\phi\right)\right)U_{i}\left(\mu_{W},\lambda,\beta\right)\label{eq: m6}
\end{align}
where $\theta=\left(\phi,\mu_{W},\lambda',\beta\right)'$ with $\dim\left(\theta\right)=L+J+1+J+JK+K$
as before.

The Jacobian of the moment vector equals
\[
M=\left(\begin{array}{cccccc}
M_{11} & M_{12} & M_{13} & M_{14} & M_{15} & M_{16}\\
M_{21} & M_{22} & M_{23} & M_{24} & M_{25} & M_{26}\\
M_{31} & M_{32} & M_{33} & M_{34} & M_{35} & M_{36}\\
M_{41} & M_{42} & M_{43} & M_{44} & M_{45} & M_{46}\\
M_{51} & M_{52} & M_{53} & M_{54} & M_{55} & M_{56}\\
M_{61} & M_{62} & M_{63} & M_{64} & M_{65} & M_{66}
\end{array}\right).
\]
Considering the first block of columns in $M$, we have that
\begin{align*}
\underset{L\times L}{M_{11}} & =\mathbb{H}\left(\phi_{0}\right)
\end{align*}
with $\mathbb{H}\left(\phi_{0}\right)$ equal the $L\times L$ expected
Hessian matrix associated with the generalized propensity score log-likelihood
(under Assumption \ref{ass: genalize_p_score} we have that $-\mathbb{H}\left(\phi_{0}\right)=\mathbb{E}\left[\mathbb{S}\mathbb{S}'\right]$).
We also have that
\[
\underset{J\times L}{M_{21}}=0,\thinspace\thinspace\underset{1\times L}{M_{31}}=0,\thinspace\thinspace\underset{J\times L}{M_{41}}=0,\thinspace\thinspace\underset{JK\times L}{M_{51}}=0,
\]
and, finally
\begin{align*}
\underset{K\times L}{M_{61}}= & \mathbb{E}\left[\left(\left[\begin{array}{cc}
\begin{array}{c}
-v\left(W,\phi_{0}\right)^{-1}\frac{\partial v\left(W,\phi_{0}\right)}{\partial\phi_{1}}v\left(W,\phi_{0}\right)^{-1}\left(X-e\left(W,\phi_{0}\right)\right)\end{array} & \cdots\end{array}\right.\right.\right.\\
 & \left.\cdots-v\left(W,\phi_{0}\right)^{-1}\frac{\partial v\left(W,\phi_{0}\right)}{\partial\phi_{L}}v\left(W,\phi_{0}\right)^{-1}\left(X-e\left(W,\phi_{0}\right)\right)\right]\\
 & \left.\left.+v\left(W,\phi_{0}\right)^{-1}\frac{\partial e\left(W,\phi_{0}\right)}{\partial\phi'}\right)U\left(\mu_{W},\lambda_{*},\beta_{0}\right)\right].
\end{align*}
Iterated expectations gives
\begin{equation}
M_{61}=\mathbb{E}\left[\left[\begin{array}{ccc}
c_{1}\left(W,\phi_{0}\right) & \cdots & c_{L}\left(W,\phi_{0}\right)\end{array}\right]\mathbb{C}\left(\left.X,U_{*}\right|W\right)+d\left(W,\phi_{0}\right)\mathbb{E}\left[\left.U_{*}\right|W\right]\right]\label{eq: M_61_expression_1}
\end{equation}
with $c_{l}\left(W,\phi\right)=-v\left(W,\phi\right)^{-1}\frac{\partial v\left(W,\phi\right)}{\partial\phi_{l}}v\left(W,\phi\right)^{-1}$
for $l=1,\ldots,L$ and $d\left(W,\phi\right)=v\left(W,\phi\right)^{-1}\frac{\partial e\left(W,\phi\right)}{\partial\phi'}$. 

It is useful to develop an alternative expression for (\ref{eq: M_61_expression_1}).
Note that
\[
\mathbb{E}\left[m_{5}(Z_{i},\phi_{0},\mu_{W},\lambda_{*},\beta_{0})\right]=\mathbb{E}\left[v\left(W,\phi_{0}\right)^{-1}\left(X-e\left(W,\phi_{0}\right)\right)U\left(\mu_{W},\lambda_{*},\beta_{0}\right)\right]=0,
\]
is mean zero. A GIME argument, similar to the one used to derive (\ref{eq: GIME})
in the main text, therefore gives
\begin{equation}
\mathbb{E}\left[\frac{\partial}{\partial\phi'}\left\{ v\left(W,\phi_{0}\right)^{-1}\left(X-e\left(W,\phi_{0}\right)\right)\right\} U_{*}\right]=-\mathbb{E}\left[v\left(W,\phi_{0}\right)^{-1}\left(X-e\left(W,\phi_{0}\right)\right)U_{*}\mathbb{S}'\right],\label{eq: GIME_app}
\end{equation}
where we use the fact that $U_{*}=U\left(\mu_{W},\lambda_{*},\beta_{0}\right)$
does not vary with the propensity score parameter, $\phi$. We can
use (\ref{eq: GIME_app}) to write
\[
M_{61}=-\mathbb{E}\left[v\left(W,\phi_{0}\right)^{-1}\left(X-e\left(W,\phi_{0}\right)\right)U_{*}\mathbb{S}'\right].
\]

If both Assumptions \ref{ass: genalize_p_score} and \ref{ass: clp_coefficients}
hold simultaneously, then $U_{*}=U_{0}$ is conditionally mean zero
and uncorrelated with $X$ (i.e., $\mathbb{E}\left[\left.U_{0}\right|W\right]=\mathbb{E}\left[\left.XU_{0}\right|W\right]=0$).
In this case $M_{61}=0$ (see Equation (\ref{eq: M_61_expression_1})
above). If Assumption \ref{ass: clp_coefficients} does not hold,
then $M_{61}$ may be non-zero.

Turning to the second block of columns in $M$, we have that
\[
\underset{L\times J}{M_{12}}=0,\thinspace\thinspace\underset{J\times J}{M_{22}}=-I_{J},
\]
and also that
\begin{align*}
\underset{1\times J}{M_{32}}= & \mathbb{E}\left[\left\{ \gamma_{*}+\left(I_{J}\otimes X\right)^{^{\prime}}\delta_{*}\right\} '\right]\\
\underset{J\times J}{M_{42}}= & \mathbb{E}\left[-I_{J}U_{*}+\left(W-\mu_{W}\right)\left\{ \gamma_{*}+\left(I_{J}\otimes X\right)^{^{\prime}}\delta_{*}\right\} '\right]\\
= & \mathbb{E}\left[\left(W-\mu_{W}\right)\left\{ \left(I_{J}\otimes X\right)^{^{\prime}}\delta_{*}\right\} '\right]\\
\underset{JK\times J}{M_{52}}= & \mathbb{E}\left[-\left(I_{J}\otimes X\right)U_{*}+\left(\left(W-\mu_{W}\right)\otimes X\right)\left(\gamma_{*}+\left(I_{J}\otimes X\right)^{^{\prime}}\delta_{*}\right)^{^{\prime}}\right]\\
\overset{A.\ref{ass: clp_coefficients}}{=} & \mathbb{E}\left[\left(\left(W-\mu_{W}\right)\otimes X\right)\left(\gamma_{0}+\left(I_{J}\otimes X\right)^{^{\prime}}\delta_{0}\right)^{^{\prime}}\right].
\end{align*}
Note that the second equality after $M_{42}$ does not require Assumption
\ref{ass: clp_coefficients} to hold. Even if $\lambda_{*}$ does
not correctly parameterize the CLP coefficients, it remains true that
$U\left(\mu_{W},\lambda_{*},\beta_{0}\right)$ is mean zero. However
$U\left(\mu_{W},\lambda_{*},\beta_{0}\right)$ may covary with $X$
when Assumption \ref{ass: clp_coefficients} fails. Therefore the
second equality after $M_{52}$ \emph{does} require Assumption \ref{ass: clp_coefficients}
to hold. The forms of $M_{32}$, $M_{42}$ and $M_{52}$ determine
the effect of sampling uncertainty about the value of $\mu_{W}$ on
sampling uncertainty about the value of $\beta_{0}$. 

Finally we get
\begin{align*}
\underset{K\times J}{M_{62}} & =\mathbb{E}\left[v\left(W;\phi_{0}\right)^{-1}\left(X-e\left(W;\phi_{0}\right)\right)\left\{ \gamma_{*}+\left(I_{J}\otimes X\right)^{^{\prime}}\delta_{*}\right\} '\right]
\end{align*}

Turning to the third block of columns in $M$, we have that
\[
\underset{L\times1}{M_{13}}=0,\thinspace\thinspace\underset{J\times1}{M_{23}}=0,\thinspace\thinspace\underset{1\times1}{M_{33}}=-1,
\]
and also that
\[
\underset{J\times1}{M_{43}}=-\mathbb{E}\left[\left(W-\mu_{W}\right)\right]=0,\thinspace\thinspace\underset{JK\times1}{M_{53}}=-\mathbb{E}\left[\left(\left(W-\mu_{W}\right)\otimes X\right)\right]
\]
and
\[
\underset{K\times1}{M_{63}}=-\mathbb{E}\left[v\left(W;\phi_{0}\right)^{-1}\left(X-e\left(W;\phi_{0}\right)\right)\right]=0.
\]

Turning to the fourth block of columns in $M$, we have that
\[
\underset{L\times J}{M_{14}}=0,\thinspace\thinspace\underset{J\times J}{M_{24}}=0,\thinspace\thinspace\underset{1\times J}{M_{34}}=-\mathbb{E}\left[\left(W-\mu_{W}\right)^{^{\prime}}\right]=0,
\]
and also that
\[
\underset{J\times J}{M_{44}}=-\mathbb{E}\left[\left(W-\mu_{W}\right)\left(W-\mu_{W}\right)'\right]=-\Sigma_{WW},\thinspace\thinspace\underset{JK\times J}{M_{54}}=-\mathbb{E}\left[\left(\left(W-\mu_{W}\right)\otimes X\right)\left(W-\mu_{W}\right)^{^{\prime}}\right],
\]
and finally that
\[
\underset{K\times J}{M_{64}}=-\mathbb{E}\left[v\left(W;\phi_{0}\right)^{-1}\left(X-e\left(W;\phi_{0}\right)\right)\left(W-\mu_{W}\right)^{^{\prime}}\right]=0.
\]

Turning to the fifth block of columns in $M$, we have that
\[
\underset{L\times JK}{M_{15}}=0,\underset{J\times JK}{M_{25}}=0,
\]
and also that
\[
\underset{1\times JK}{M_{35}}=-\mathbb{E}\left[\left(\left(W-\mu_{W}\right)\otimes X\right)^{^{\prime}}\right],\thinspace\thinspace\underset{J\times JK}{M_{45}}=-\mathbb{E}\left[\left(W-\mu_{W}\right)\left(\left(W-\mu_{W}\right)\otimes X\right)^{^{\prime}}\right],
\]
and also that
\[
\underset{JK\times JK}{M_{55}}=-\mathbb{E}\left[\left(\left(W-\mu_{W}\right)\otimes X\right)\left(\left(W-\mu_{W}\right)\otimes X\right)^{^{\prime}}\right],
\]
and finally that
\[
\underset{K\times JK}{M_{65}}=-\mathbb{E}\left[v\left(W;\phi_{0}\right)^{-1}\left(X-e\left(W;\phi_{0}\right)\right)\left(\left(W-\mu_{W}\right)\otimes X\right)^{^{\prime}}\right]=0.
\]

Turning to the sixth, and final, block of columns in $M$, we have
that
\[
\underset{L\times K}{M_{16}}=0,\thinspace\thinspace\underset{J\times K}{M_{26}}=0,
\]
and also that
\[
\underset{1\times K}{M_{36}}=-\mathbb{E}\left[X^{^{\prime}}\right],\thinspace\thinspace\underset{J\times K}{M_{46}}=-\mathbb{E}\left[\left(W-\mu_{W}\right)X^{^{\prime}}\right],\thinspace\thinspace\underset{JK\times K}{M_{56}}=-\mathbb{E}\left[\left(\left(W-\mu_{W}\right)\otimes X\right)X^{^{\prime}}\right],
\]
and finally that
\[
\underset{K\times K}{M_{66}}=-\mathbb{E}\left[v\left(W;\phi_{0}\right)^{-1}\left(X-e\left(W;\phi_{0}\right)\right)X^{^{\prime}}\right]=-I_{K}.
\]

Marking out the zero and identity terms we have that
\[
M=\left(\begin{array}{cccccc}
M_{11} & 0 & 0 & 0 & 0 & 0\\
0 & -I_{J} & 0 & 0 & 0 & 0\\
0 & M_{32} & -1 & 0 & M_{35} & M_{36}\\
0 & M_{42} & 0 & M_{44} & M_{45} & M_{46}\\
0 & M_{52} & M_{53} & M_{54} & M_{55} & M_{56}\\
M_{61} & M_{62} & 0 & 0 & 0 & -I_{K}
\end{array}\right).
\]
Using the above we have (\ref{eq: B1_matrix}), as defined in the
appendix to the main paper, equal to
\[
\underset{\left(1+J+JK+K\right)\times J}{B_{1}}=\left(\begin{array}{c}
M_{32}\\
M_{42}\\
M_{52}
\end{array}\right)=\mathbb{E}\left[\begin{array}{c}
\left\{ \gamma_{*}+\left(I_{J}\otimes X\right)^{^{\prime}}\delta_{*}\right\} '\\
\left(W-\mu_{W}\right)\left\{ \left(I_{J}\otimes X\right)^{^{\prime}}\delta_{*}\right\} '\\
-\left(I_{J}\otimes X\right)U_{*}+\left(\left(W-\mu_{W}\right)\otimes X\right)\left(\gamma_{*}+\left(I_{J}\otimes X\right)^{^{\prime}}\delta_{*}\right)^{^{\prime}}
\end{array}\right].
\]

\subsection*{Additional detailed calculations}

\subsubsection*{Equation (\ref{eq: inverse_Jacobian})}

To derive the lower-left-hand block of (\ref{eq: inverse_Jacobian})
in the Appendix to the main paper we multiply out:
\begin{align*}
-\left(\begin{array}{cc}
\mathbb{E}\left[RR'\right]^{-1} & -\mathbb{E}\left[RR'\right]^{-1}\mathbb{E}\left[RX'\right]\\
0 & I_{K}
\end{array}\right)\left(\begin{array}{cc}
0 & -B_{1}\\
-M_{61} & -B_{2}
\end{array}\right)\left(\begin{array}{cc}
-\mathbb{H}\left(\phi_{0}\right)^{-1} & 0\\
0 & I_{J}
\end{array}\right) & =\\
-\left(\begin{array}{cc}
\mathbb{E}\left[RR'\right]^{-1}\mathbb{E}\left[RX'\right]M_{61} & -\mathbb{E}\left[RR'\right]^{-1}\left(B_{1}-B_{2}\mathbb{E}\left[RX'\right]\right)\\
-M_{61} & -B_{2}
\end{array}\right)\left(\begin{array}{cc}
-\mathbb{H}\left(\phi_{0}\right)^{-1} & 0\\
0 & I_{J}
\end{array}\right) & =\\
-\left(\begin{array}{cc}
-\mathbb{E}\left[RR'\right]^{-1}\mathbb{E}\left[RX'\right]M_{61}\mathbb{H}\left(\phi_{0}\right)^{-1} & -\mathbb{E}\left[RR'\right]^{-1}\left(B_{1}-B_{2}\mathbb{E}\left[RX'\right]\right)\\
M_{61}\mathbb{H}\left(\phi_{0}\right)^{-1} & -B_{2}
\end{array}\right).
\end{align*}

\subsubsection*{Equation (\ref{eq: lambda_BLP})}

To derive (\ref{eq: lambda_BLP}) in the Appendix start by observing
that moment (\ref{eq: m3_p1}) in the main text implies the following
characterization of $\lambda_{0}$ and $\beta_{0}:$
\begin{align*}
\left[\begin{array}{c}
\mathbb{E}\left[RY\right]\\
\mathbb{E}\left[v_{0}\left(W\right)^{-1}\left(X-e_{0}\left(W\right)\right)Y\right]
\end{array}\right] & =\left[\begin{array}{cc}
\mathbb{E}\left[RR'\right] & \mathbb{E}\left[RX'\right]\\
0 & I_{K}
\end{array}\right]\left(\begin{array}{c}
\lambda_{0}\\
\beta_{0}
\end{array}\right).
\end{align*}
After some calculation we get that
\begin{align*}
\lambda_{0} & =\mathbb{E}\left[RR'\right]^{-1}\mathbb{E}\left[R\left(Y-X'\beta_{0}\right)\right]\\
 & =\mathbb{E}\left[RR'\right]^{-1}\mathbb{E}\left[R\left(a_{0}\left(W\right)+X'\left(b_{0}\left(W\right)-\beta_{0}\right)\right)+U_{0}\right]\\
 & =\mathbb{E}\left[RR'\right]^{-1}\mathbb{E}\left[R\left(a_{0}\left(W\right)+X'\left(b_{0}\left(W\right)-\beta_{0}\right)\right)\right],
\end{align*}
where the last line uses Lemma \ref{lem: Wooldridge} of the main
text.

\subsubsection*{Equation (\ref{eq: rewritten_key_variance})}

To derive equation (\ref{eq: rewritten_key_variance}) in the Appendix
expand the variance of $b_{0}\left(W\right)+\Delta_{*}^{\left(J\right)}\left(k^{\left(J\right)}\left(W\right)-\mu^{\left(J\right)}\right)-\beta_{0}$
as follows:
\begin{align}
\mathbb{V}\left(b_{0}\left(W\right)+\Delta_{*}^{\left(J\right)}\left(k^{\left(J\right)}\left(W\right)-\mu^{\left(J\right)}\right)-\beta_{0}\right)= & \Delta_{*}^{\left(J\right)}\mathbb{V}\left(k^{\left(J\right)}\left(W\right)\right)\left(\Delta_{*}^{\left(J\right)}\right)'\nonumber \\
 & +\mathbb{V}\left(b_{0}\left(W\right)\right)\nonumber \\
 & +2\mathbb{E}\left[\left(b_{0}\left(W\right)-\beta_{0}\right)\left\{ \Delta_{*}^{\left(J\right)}\left(k^{\left(J\right)}\left(W\right)-\mu^{\left(J\right)}\right)\right\} '\right]\nonumber \\
= & \Delta_{*}^{\left(J\right)}\mathbb{V}\left(k^{\left(J\right)}\left(W\right)\right)\left(\Delta_{*}^{\left(J\right)}\right)'\nonumber \\
 & +\mathbb{V}\left(b_{0}\left(W\right)\right)\nonumber \\
 & +2\mathbb{E}\left[\left(b_{0}\left(W\right)-\beta_{0}\right)\right.\nonumber \\
 & \times\left.\left\{ b_{0}\left(W\right)+\Delta_{*}^{\left(J\right)}\left(k^{\left(J\right)}\left(W\right)-\mu^{\left(J\right)}\right)-\beta_{0}\right.\right]\nonumber \\
 & \times\left.\left.-\left(b_{0}\left(W\right)-\beta_{0}\right)\right\} '\right]\nonumber \\
= & \Delta_{*}^{\left(J\right)}\mathbb{V}\left(k^{\left(J\right)}\left(W\right)\right)\left(\Delta_{*}^{\left(J\right)}\right)'-\mathbb{V}\left(b_{0}\left(W\right)\right)\nonumber \\
 & -2\mathbb{E}\left[\left(b_{0}\left(W\right)-\beta_{0}\right)\right.\nonumber \\
 & \left.\times\left\{ b_{0}\left(W\right)+\Delta_{*}^{\left(J\right)}\left(k^{\left(J\right)}\left(W\right)-\mu^{\left(J\right)}\right)-\beta_{0}\right\} '\right].\label{eq: clever_variance_manipulation}
\end{align}

\end{document}